\documentclass[review]{elsarticle}

\usepackage{amssymb}
\usepackage[figuresright]{rotating}
\usepackage{multirow}
\usepackage{algorithm}
\usepackage{algorithmicx}
\usepackage{algpseudocode}
\usepackage{subfigure}
\usepackage{amsmath}
\usepackage{amsthm}
\usepackage{amsfonts}
\usepackage{bbm}
\usepackage{color}
\usepackage{url}

\newtheorem{theorem}{Theorem}
\newtheorem{lemma}{Lemma}
\newtheorem{definition}{Definition}
\newtheorem{example}{Example}
\usepackage{lineno}
\journal{Information Sciences}

\begin{document}

\begin{frontmatter}



\title{Local Differentially Private Frequency Estimation based on Learned Sketches}


\author[]{Meifan Zhang$^a$, Sixin Lin$^a$$^,$$^b$, Lihua Yin$^a$$^*$}

\affiliation{organization={Cyberspace Institute of Advanced Technology},%Department and Organization
            addressline={Guangzhou University},
            city={Guangzhou},
            postcode={510006},
            country={China}}

\affiliation{organization={Peng Cheng Laboratory},%Department and Organization
            addressline={},
            city={Shenzhen},
            postcode={518000},
            country={China}}

\cortext[]{Meifan Zhang and Sixin Lin should be considered joint first author.}

\begin{abstract}
Sketches are widely used for frequency estimation of data with a large domain. However, sketches-based frequency estimation faces more challenges when considering privacy. Local differential privacy (LDP) is a solution to frequency estimation on sensitive data while preserving the privacy. LDP enables each user to perturb its data on the client-side to protect the privacy, but it also introduces errors to the frequency estimations. The hash collisions in the sketches make the estimations for low-frequent items even worse. In this paper, we propose a two-phase frequency estimation framework for data with a large domain based on an LDP learned sketch, which separates the high-frequent and low-frequent items to avoid the errors caused by hash collisions. We theoretically proved that the proposed method satisfies LDP. Our method is more accurate than the state-of-the-art frequency estimation methods under LDP including Apple-CMS, Apple-HCMS and FLH. The experimental results verify the performance of our method.
\end{abstract}




\end{frontmatter}



\section{Introduction}
Frequency estimation is a traditional and important problem in data analytics.
At present, to solve the frequency estimation problem, it is necessary to consider not only the efficiency and space cost but also the risk of privacy leakage. In the process of collecting data from users and aggregating the frequencies, some sensitive information may leak. For example, when the studies do frequency statistics on some diseases or medicines, the data providers do not want to reveal their true illnesses or medications. Thus, we need to protect users' privacy while estimating the frequencies.

However, privacy-preserving for big data frequency estimation is not an easy task.
Conventional encryption techniques are too costly for big data, thus, they cannot meet the requirement of fast response for big data analytics. In the meantime, the data owners often use generalization and suppression techniques to achieve anonymization requirements and protect their data privacy~\cite{sweeney2002k}. But the disadvantage of anonymization techniques is that they do not provide a measure of privacy loss in big data analysis. In recent years, differential privacy (DP)~\cite{dwork2008differential} is a popular privacy-preserving solution due to its strong mathematical boundary of the leaked privacy. It adds noises to the aggregations to avoid the leakage of individual privacy. But it is difficult to find a trusted third party to aggregate the data from a large number of clients. Local differential privacy (LDP)~\cite{kasiviswanathan2011can} is a solution to this problem, which locally perturbs the raw data before sending it to the server. In this way, the server has no access to the raw data, thus, the privacy of each client is protected. Many companies, such as the Google~\cite{erlingsson2014rappor,fanti2015building}, Apple~\cite{APPLE}, and Microsoft~\cite{ding2017collecting} adopt LDP to collect and aggregate sensitive data from users.

However, frequency estimation for big data under LDP still faces some challenges.
On the one hand, it is difficult to get sufficiently accurate frequency estimation for big data with a large domain. Methods such as the Basic RAPPOR~\cite{fanti2015building}, OUE~\cite{wang2017locally}, and OLH~\cite{wang2017locally} can not handle the data with a large domain. Sketches-based methods use hash functions to map the data with a large domain to a sub-linear space. They are tailored for streaming data analysis on architectures even with limited memory such as single-board computers that are widely exploited for IoT and edge computing~\cite{yildirim2020differentially}. The methods such as RAPPOR~\cite{erlingsson2014rappor} and Apple-CMS~\cite{APPLE} use the bloom-filters or sketches to reduce the domain size, but the accumulated errors due to hash collisions reduces the estimation accuracy as the data grows rapidly. Separating the storage of high-frequent and low-frequent items is a way to avoid hash collisions. But it is difficult for the server to distinguish them since the perturbed values of different items are sufficiently similar according to LDP.
On the other hand, it is challenging to separate the storage of high-frequent items and low-frequent items while preserving privacy. Both the server and the clients have the risk of privacy leakage when treating the high-frequent and low-frequent items in different ways. Thus, methods separating the items by their frequencies while remaining their privacy are required.

To tackle the first challenge, we try to improve the accuracy of frequency estimation under LDP by avoiding the collisions between the high-frequent items and low-frequent items. However, an individual user with no prior knowledge cannot identify whether an item is high-frequent or not. Instead of enabling the server to distinguish the perturbed values, we try to enable each client to identify whether its item is a high-frequent one or not. We train a frequency model based on the aggregations of the perturbed values from some sample clients. The clients other than those in the samples can use the model to distinguish whether their items are high-frequent ones or not. Since the frequency of a high-frequent item can be accurately estimated based on the samples, we use the model to replace the storage of high-frequent items and leave the sketch for the low-frequent items. Thus, the hash collisions between high-frequent items and low-frequent items can be reduced.

To tackle the second challenge, we let the clients encode the high-frequent items and low-frequent items in different ways while satisfying the LDP. A naive idea to reduce the hash collisions between high-frequent items and low-frequent items is to let the clients only send the low-frequent items to the server since the high-frequent items can be predicted according to the frequency model. However, if the client does not send information to the server, it reveals that its item a high-frequent one. To avoid this privacy leakage, the clients with high-frequent items should also send some values to the server, and these values must be difficult to distinguish from the perturbed values of the low-frequent items. To avoid the errors caused by involving the perturbed high-frequent items in the sketch, we propose a method to make the perturbed values of high-frequent items uniformly disperse in the sketch, so that we can accurately evaluate and eliminate the impact of these items from the estimations. We prove that this method satisfies LDP and reduces the variances of estimations for low-frequent items.

\begin{figure}
\centering
\includegraphics[scale=0.4]{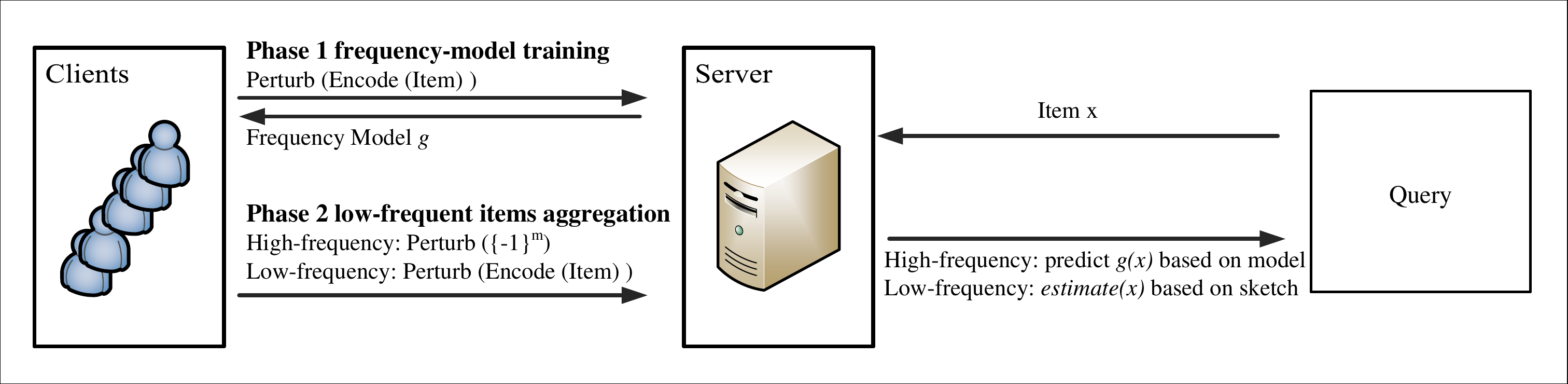}
\caption{The idea of LDPLCM framework.}
\label{Fig:intro_framework}
\end{figure}

In this paper, we propose a two-phase LDP frequency estimation algorithm as shown in Figure~\ref{Fig:intro_framework}. The first phase trains a frequency model based on the aggregations of the perturbed values from some sample clients. The model enables each client to distinguish high-frequent items and low-frequent items. It also enables the server to estimate the frequencies of high-frequent items by the predictions according to the model.
In the second phase, each client uses the frequency model to identify whether its item is high-frequent or not. The clients encode the high-frequent and low-frequent items in different ways while satisfying LDP. The server uses a sketch to aggregates the perturbed data from clients and estimate the frequencies of low-frequent items.

The main contributions of this work are summarized as follows:
\begin{itemize}
\item We present a local differentially private frequency estimation framework based on learned sketches (LDPLCM). It can estimate the frequency of sensitive data with a large domain and guarantee the high utility of estimations.
\item We proved that the proposed method satisfies LDP. It can treat the high-frequent items and low-frequent items in different ways to avoid the hash-collisions while preserving the privacy.
\item We proved that the proposed method is more accurate for low-frequent items than the LDP frequency estimation method Apple-CMS.
\item We conduct extensive experiments on both synthetic and real-world datasets. The experimental results show that the proposed method outperforms the state-of-the-art LDP frequency estimation algorithms including Apple-CMS, Apple-HCMS and FLH.

\end{itemize}
The remainder of this paper is organized as follows. In section 2, we survey the related work for this paper. In section 3, we introduce some preliminaries. Section 4 presents the LDPLCM algorithm. In section 5, the experimental results show the performance of the proposed algorithm. In section 6, we conclude the paper with future directions.

\section{Related Works}

Frequency estimation can be applied in many fields such as finding frequent items~\cite{li2020wavingsketch,karp2003simple,tai2018sketching}, hierarchical heavy hitters~\cite{cormode2003finding,tang2016graph}, network measurements~\cite{yang2019adaptive,liu2016one}. In ~\cite{basat2019randomized}, the authors use frequency estimation and top-k items identification to perform network monitoring.

Differential privacy is extensively studied for protecting users' privacy while enabling big data analysis. Frequency estimation of sensitive data poses a risk of privacy leakage. For example, users report their symptoms to the disease control department through their mobile phones in the medical Internet of Things. There has a risk of privacy leakage in the report because the third-party data recipients are not completely trusted. The central model of DP also faces the same problem that no trusted third-party server can be found. Thus, LDP~\cite{kasiviswanathan2011can} has been implemented to protect clients' privacy in frequency estimation, such as RAPPOR~\cite{erlingsson2014rappor}, Apple-CMS~\cite{APPLE}, Apple-HCMS~\cite{APPLE}, pure LDP~\cite{wang2017locally}, k-RR~\cite{kairouz2016discrete}. In~\cite{erlingsson2014rappor}, Erlingsson et.al combine randomized response with Bloom filters~\cite{bloom1970space} to satisfy $\epsilon$-LDP~\cite{duchi2013local}. They use the permanent random response instead of clients' initial data and calculate the instantaneous randomized response to perturb the permanent random response. Their extending work from RAPPOR is to learn the joint distributions and associations between unknown data dictionaries~\cite{fanti2015building}. These ways enhance the difficulty of tracking clients' activity for attackers, but they complicate the decoding process on the server. To this end, Apple-CMS and Apple-HCMS estimate the frequencies based on the perturbed data items in sketches and directly calculate the average value from hash entries without decoding. The advantage of Apple-HCMS over Apple-CMS is the reduction of communication costs. These two methods can be used for data in a large domain but they fail to decrease the estimation errors caused by the hash collisions.

Finding the trade-off between privacy budget and accuracy can improve the data utility in frequency estimation~\cite{wang2019locally,wang2019answering,murakami2019utility,wei2020asgldp,xu2020collecting}. Kairouz et.al~\cite{kairouz2016discrete} propose and prove that the hashed k-RR is optimal in the low privacy regime compared with RAPPOR. Wang et.al~\cite{wang2017locally} generalize the RAPPOR and design three perturbation methods including Direct Encoding(DE), Optimized Unary Encoding(OUE), Optimized Local Hashing(OLH) in the pure LDP framework. They prove an unbiased estimate and find the best parameters to minimize the variance of estimation. In~\cite{murakami2019utility}, Murakami et.al consider the data sensitivity of clients and propose the utility-optimized LDP(ULDP) mechanisms to maximize the utility. Jia et.al~\cite{jia2019calibrate} associate the prior knowledge with LDP in Calibrate framework to count the true items, and model the distribution probability from prior knowledge. To improve the efficiency, Flash Local Hashing(FLH) and Hadamard Response(HR)~\cite{cormode2021frequency} restrict the clients' choices of $k'$ hash functions. They are faster than OLH and Apple-HCMS by introducing a matrix $k' \times m (k' \ll n)$ to store the perturbed data items, which are suitable for small data domains.

Sketches have been refined over the past two decades as common tools for frequency estimation on data with a large domain. Many typical sketches are proposed including Count-Sketches~\cite{charikar2002finding}, Count-Min sketches~\cite{cormode2005improved}, Augmented sketches~\cite{roy2016augmented}, CU sketches~\cite{estan2002new}. These sketches have common structures, they use hash functions to map a large amount of data to a two-dimensional array to reduce the space cost. But they also face the errors caused by hash collisions, the hash collisions between high-frequent items and low-frequent items introduce great errors for the estimation of low-frequent items. To reduce the errors, many methods are proposed to distinguish the high-frequent items from the low-frequent ones. Augmented sketches~\cite{roy2016augmented} add a filter on the top of the Count-Min sketch to store the frequencies of top-k items. The cold filter~\cite{zhou2018cold} and HeavyGuardian~\cite{yang2018heavyguardian} separate the cold items and hot items into two stages.

However, separating the high-frequent items from the low-frequent items is non-trivial under LDP due to the privacy-preserving. The LDPMiner~\cite{qin2016heavy} is a two-phase heavy hitter mining algorithm. It gathers a candidate set of heavy hitters in the first phase and estimates the frequencies of these candidates in the second phase. But it only focuses on the frequency estimation for the top-k frequent items. Sending the frequency property of each item to the clients is costly for data with a large domain. Inspired by the learned index~\cite{kraska2018case}, some learning-based frequency estimation methods are proposed in recent years~\cite{hsu2019learning, zhang2020learned, zhou2019rl}. We also attempt to use a lightweight model to learn the frequency properties of the sensitive data from the clients. However, the learning-based sketch cannot be directly applied to the frequency estimation under LDP. We need to separate the high-frequent items and low-frequent items while preserving privacy. In this paper, we propose a two-phase frequency estimation algorithm based on the learned sketch under LDP. In the first phase, we train a lightweight model to predict the frequency of each item based on the aggregations of the perturbed values from some sample clients. In the second phase, the rest of the clients identify and encode the high-frequent and low-frequent items in different ways according to the frequency model. The proposed method satisfies LDP, and it is more accurate than the state-of-the-art sketch-based frequency estimation method Apple-CMS.

%How to improve the efficiency and accuracy with the requirements of the LDP framework is still researching. In this paper, we use a two-phase structure for estimation. In phase1, we sampling and use machine learning methods instead Apple-CMS to train the model for high-frequent items. In phase2, we distinguish the items according to the threshold and the low-frequent items are calculated by sketch.

\section{Preliminaries}
In this section, we provide some preliminaries including LDP and sketches for frequency estimation.
\subsection{Local Differential Privacy}
DP has been accepted as the de facto standard for data privacy. But DP is not applicable when there is no trusted aggregator. LDP is proposed to handle this obstacle. In the local setting for DP, there are a large number of users and one aggregator. To protect the privacy of individual users, each user locally perturbs its private data and sends the perturbed value to the aggregator.

\begin{definition}
($\epsilon$-Local Differential Privacy). An algorithm $A(\cdot)$ satisfies $\epsilon$-LDP if and only if for $\epsilon>0$ and any inputs $v_1, v_2\in D$ from the dataset $D$, we have
\begin{equation}\label{Def:LDP}
  \forall T\in A(D):Pr[A(v_1)\in T]\le \mathrm{e}^{\epsilon} Pr[A(v_2)\in T],
\end{equation}
where $A(D)$ denotes the set of all possible outputs of the algorithm $A$.
\end{definition}

%\textcolor{blue}{
%\begin{definition}
%(($\epsilon$, $\delta$)-Local Differential Privacy). An algorithm $A(\cdot)$ satisfies ($\epsilon$, $\delta$)-LDP if and only if for the $\epsilon>0$ and all the possible pairs of input $v_1, v_2\in D$, we have
%\begin{equation}\label{Def:LDP}
%  \forall T\in Range(A):Pr[A(v_1)\in T]\le \mathrm{e}^{\epsilon} Pr[A(v_2)\in T] + \delta,
%\end{equation}
%where $\delta$ is typically small, Range($A$) denotes the set of all possible outputs of the algorithm $A$. Generally speaking, ($\epsilon$, $\delta$)-LDP means that an algorithm A achieves $\epsilon$-LDP with probability at least $1-\delta$. Actually, ($\epsilon$, $\delta$)-LDP is more general since the latter in the special case of $\delta$ = 0 becomes the former.
%\end{definition}}

LDP ensures that the outputs of the random algorithm with different inputs are similar enough, thus the perturbed values will not leak the privacy of the inputs.
\subsection{Sketches}
Sketches support count queries over data with a large domain, and they summarize a large amount of data into sub-linear space. In recent years, many typical sketches such as the count sketches~\cite{charikar2002finding} and the count-min sketches~\cite{cormode2005improved} are proposed to provide more accurate estimations. The Apple-CMS and Apple-HCMS algorithms are also designed based on the count-min sketches to provide the frequency estimation under LDP. Therefore, we review these two state-of-the-art sketches including the count-min sketch and the count sketch.

\subsubsection{Count-Min Sketch}

A Count-Min sketch~\cite{cormode2005improved} with parameters $(\epsilon_{CM},\delta_{CM})$ is represented by a two-dimensional array counts with width $m$ and depth $k$: $count[1,1]$,...,$count[k,m]$. Given parameters $(\epsilon_{CM},\delta_{CM})$, set $m=\lceil e/\epsilon_{CM}\rceil$ and $k=\lceil ln(1/\delta_{CM})\rceil$. Each cell of the array is initialised with a zero. The $d$ hash functions $h_1,..., h_d:\{1,...,n\}\rightarrow\{1,...,m\}$ are used to update the count in each cell of the array. The parameters $(\epsilon_{CM},\delta_{CM})$ have nothing to do with the parameter $\epsilon$ of DP. We add subscripts ``CM'' to the parameters of the Count-Min sketch to distinguish them.
Count-Min sketches has two basic operations: \emph{update} and \emph{estimate}.

Update($x$):
\begin{equation}\label{Equation:Count-Min_Update}
  count[i,h_i(x)]\leftarrow count[i,h_i(x)]+1
\end{equation}

Estimation($x$):
\begin{equation}\label{Equation:Count-Min_Estimation}
  min_{i\in[1,k]}\{count[i,h_i(x)]\}
\end{equation}

This is an over-estimation due to the hash collisions $\hat{a_i}\ge a_i$, where $a_i$ is the count of the item $x$ and $\hat{a_i}$ is the estimation. The estimation has an upper bound, with the probability at least (1-$\delta_{CM}$), $\hat{a_i}\le a_i + \epsilon_{CM} \Vert A\Vert_1$, where $\Vert A\Vert_1=\sum_{i=1}^{n}|a_i|$~\cite{cormode2005improved}.

The Apple-CMS~\cite{APPLE} adopts a variance of Count-Min sketch to encode and aggregate the sensitive data. It encodes the value into a one-hot vector with a hash function randomly chosen from the $h_1,... ,h_k$. It then perturbs the vector and sends the perturbed vector to the server. The server adds the vector to the corresponding line of the sketch. It estimates the frequency of an item $x$ with the appropriate correction of $mean_{i\in[1,k]}\{count[i,h_i(x)]\}$.

\subsubsection{Count Sketch}
%A count sketch with parameters $(\epsilon,\delta)$ is represented by a two-dimensional array counts with width $w$ and depth $d$: $count[1,1]$,...,$count[d,w]$. Given parameters $(\epsilon,\delta)$, set $w=\lceil e/\epsilon\rceil$ and $d=\lceil ln(1/\delta)\rceil$. Each cell of the array is initialised with a zero.
The Count Sketch~\cite{charikar2002finding} has the same structure with the count-min sketch. The only difference between them is that the count sketch adds another hash function $s_i$ mapping each item to $\{-1,1\}$ for each line of the array. %The sketch uses $d$ hash functions $h_1$,... , $h_d:\{1,...,n\}\rightarrow\{1,...,w\}$ to update the count in each cell of the array.It uses another hash function $s_i$ mapping each item to $\{-1,1\}$ for each line of the array,
That is, $s_i: x\rightarrow \{-1,1\}$. Its operations \emph{update} and \emph{estimate} are as follows.

Update($x$):
\begin{equation}\label{Equation:Count_Update}
  count[i,h_i(x)]\leftarrow count[i,h_i(x)]+s_i(x)
\end{equation}

Estimate($x$):
\begin{equation}\label{Equation:Count_Estimation}
  mean_{i\in[1,d]}\{count[i,h_i(x)]\cdot s_i(x)\}
\end{equation}

\section{LDP Learned Count-Mean Sketches (LDPLCM) for Frequency Estimation}

\subsection{The Framework of LDPLCM}

In this section, we will introduce the framework of the proposed method LDPLCM.

The hash collisions between high-frequent items and low-frequent items largely increase the estimation errors, especially when the data domain is very large. Therefore, separating the high-frequent items and low-frequent items is a way to reduce the hash collisions, which is adopted in many sketches-based frequency estimation methods. However, separating the high-frequent and low-frequent items is non-trivial under LDP. Each client under LDP does not know whether its value is a high-frequent one or not.

\begin{figure}
\centering
\includegraphics[scale=0.5]{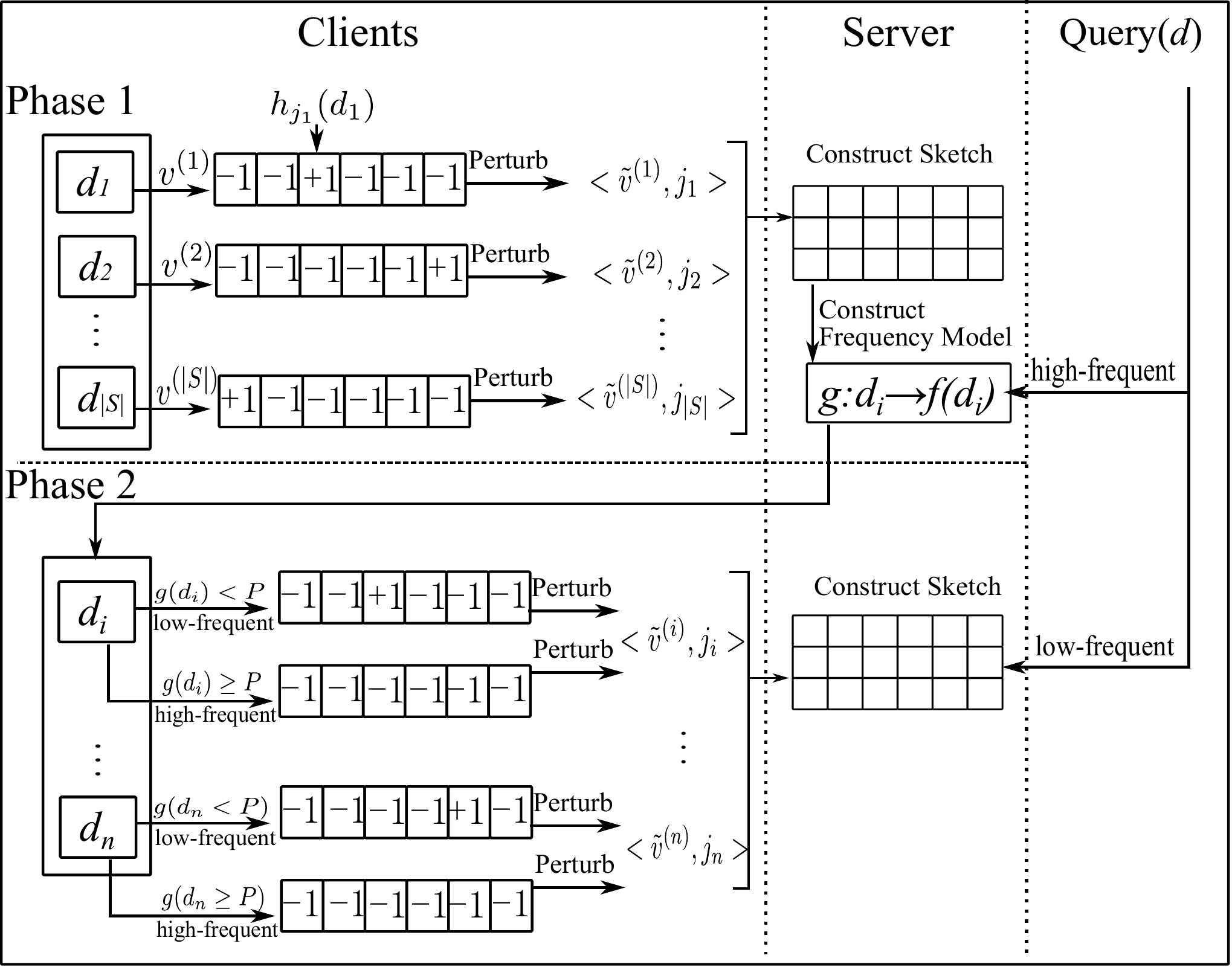}
\caption{The framework of LDPLCM.}
\label{Fig:Framework}
\end{figure}

Our main point is to reduce the errors caused by the hash collisions between the high-frequent items and low-frequent items while satisfying LDP. We will first introduce the proposed framework and then detail how to achieve the goal of reducing errors by solving a series of tasks.

The framework of our LDPLCM algorithm is shown in Figure~\ref{Fig:Framework}. In the first phase, we randomly choose some sample clients, and we encode and perturb each value of the sample clients in the same way as the Apple-CMS. The server constructs a sketch with the perturbed values and trains a frequency model $g$ mapping each data $d_i$ to its frequency $f(d_i)$ estimated by the sketch. It finds a boundary $P$ to separate the high-frequent and low-frequent items, which enables the clients in the second phase to distinguish the high-frequent items and low-frequent items.
In the second phase, the remaining clients distinguish whether their items are high-frequent or not according to the frequency model $g$ and the boundary $P$. If $g(d_i)<P$, then $d_i$ is regarded as a low-frequent item and treated in the same way as phase 1. If $g(d_i)\ge P$, then $d_i$ is regarded as a high-frequent item, and it is encoded with a vector of $\{-1\}^m$. In this way, the high-frequent items and low-frequent items are differentially encoded, but they are perturbed in the same way. The server in phase 2 empties the sketch and reconstructs it with the perturbed values from the clients. When querying the frequency of an item $d$, it still uses the frequency model to distinguish whether it is a high-frequent item or not. If it is a high-frequent item, the server predicts its frequency according to the model, otherwise, the server estimates its frequency based on the sketch.

We then detail how the framework reduces the estimation errors while preserving privacy by solving a series of tasks.

The first task is to enable each client to distinguish whether its item is a high-frequent one or not. Inspired by the LDPMiner, which finds the heavy-hitters in the first phase and estimates the frequencies of the heavy-hitters in the second phase, we also propose a two-phase algorithm. In the first phase, we attempt to learn the frequency property of items according to the aggregations of some sample clients. To avoid the significant cost of sending the frequency property of a large number of items, we train a lightweight model to predict the frequency of each item. In the second phase, we enable the clients other than the samples to distinguish the high-frequent items from the low-frequent ones according to the frequency model learned in the first phase.

The second task is to avoid the errors caused by the hash collisions between the high-frequent and low-frequent items. Even though the clients can distinguish the high-frequent items and low-frequent items according to the model, it is still non-trivial for the server to separate the storage of high-frequent and low-frequent ones. The reason is that each client perturbs its value before sending it to the server. To separate the storage of items with different frequency properties, we reuse the frequency model to replace the storage of high-frequent items and leave the sketch for the low-frequent ones. Since the model is trained based on the aggregations of the sample clients, its predictions for high-frequent items are more reliable than those for low-frequent ones. In this way, only the low-frequent items are aggregated in the sketch in the second phase, thus the hash collisions between high-frequent ones and low-frequent ones are avoided. However, it causes another problem, the client leaks information to the server whether its item is high-frequent or not if it only sends the perturbed low-frequent items to the server. At such, the client also needs to send some information to the server to avoid leaking privacy.

The third task is to protect the privacy of each client while correcting the errors caused by involving the perturbed values of high-frequent items into the sketch. This sounds contradictory, because DP works by making the probability of getting the same output from different inputs similar. In order to achieve the goal of correcting the errors, we encode all the  high-frequent items to the same vector $\{-1\}^m$. Thus, even the high-frequent items are encoded, perturbed, and sent to the server, they cause no collisions with the low-frequent items. We prove that both the encoding and perturbing satisfy LDP. Since all the high-frequent items are encoded into $\{-1\}^m$, and each bit is perturbed with the same probability, the errors caused by high-frequent items are uniformly dispersed in the sketch. Therefore, we can accurately evaluate and eliminate these errors from the estimations.

In this way, our method achieves both privacy-preserving and errors reduction. On the one hand, it protects the privacy of each client according to LDP. On the other hand, it avoids the errors caused by hash collisions between high-frequent items and low-frequent items.

\begin{algorithm}
\caption{LDPLCM}\label{Algorithm:LDPLCM}
\begin{algorithmic}[1]
\State $S\leftarrow$ sample clients with a sampling rate $r$ \Comment{Phase 1: Train the frequency model}
\For {each client $S_i$ with data $d_i$}
    \State $<\tilde{v}^{(i)}, j_i>\leftarrow$ LDPLCMclient($d_i$, $phase=1$)
    \State Send $<\tilde{v^{(i)}}, j_i>$ to Server
\EndFor
\State Server receives $DataInput\leftarrow$ $\{<\tilde{v}^{(1)}, j_1>$,  $<\tilde{v}^{(2)}, j_2>$,..., $<\tilde{v}^{(|S|)}, j_{|S|}>\}$
\State Public frequency model $g$, and boundary $P\leftarrow$ LDPLCMserver-construction($DataInput$, $phase=1$)
\For {each of the clients other than the sample clients} \Comment{Phase 2: Construct the sketch}
    \State $<\tilde{v}^{(i)}, j_i>\leftarrow$ LDPLCMclient($d_i$, $phase=2$)
    \State Send $<\tilde{v}^{(i)}, j_i>$ to Server
\EndFor
\State Server receives $DataInput\leftarrow$ $\{<\tilde{v}^{(1)}, j_1>$,  $<\tilde{v}^{(2)}, j_2>$,..., $<\tilde{v}^{(n)}, j_{n}>\}$
\State Reinitialize $M\in \{0\}^{k\times m}$
\State Get sketch $M\leftarrow$ LDPLCMserver-construction($DataInput$, $phase=2$)
\State \textbf{return: $M$, $g$, $P$}
\end{algorithmic}
\end{algorithm}

The pseudo-code of this framework is shown in Algorithm~\ref{Algorithm:LDPLCM}.
The algorithm constructs the frequency model in phase 1 (line 1-7). It first gets some sample clients (line 1). Each sample client chooses a hash function $h_{j_i}$ to encode its value and sends the perturbed value $\tilde{v}^{(i)}$ along with the index of the chosen hash function to the server (line 2-5). The server trains the frequency model $g$, and computes the frequency boundary $P$ based on the aggregations of the perturbed values from the sample clients.
The algorithm constructs the sketch in phase 2 with the clients other than the samples (line 8-14). The client-side treats high-frequent items and low-frequent items in different ways (line 9), which are introduced later. Finally, the server constructs the sketch with the perturbed values from the clients in phase 2. We will introduce the client-side algorithm LDPLCMclient and the server-side algorithm LDPLCMserver in the remaining part of this section.

\subsection{Client-side algorithm of LDPLCM}

As introduced in the framework, the client-side algorithm has two phases. In phase 1, each sample client takes the same way as $A_{client}\mbox{-}CMS$~\cite{APPLE} to encode and perturb its value. In phase 2, each of the remaining clients identifies whether its value is high-frequent or not according to the frequency model. After that, the clients encode the high-frequent items and low-frequent items in different ways to reduce the estimation errors caused by hash collisions.

The pseudo-code of the client-side algorithm is shown in Algorithm~\ref{Algorithm:LDPLCMclient}.
In phase 1, the algorithm uses the client-side algorithm of Apple-CMS~\cite{APPLE} to get the perturbed values of each sample client (line 1-2). In phase 2, if the item $d$ of the client is predicted as a low-frequent item, the client encodes and perturbs it in the same way as in phase 1 (line 4-5). If $d$ is a high-frequent item, then it is encoded as a vector $\{-1\}^m$ (line 8). The vector $v$ is then perturbed by multiplying each of its bits with (-1) with a probability $\frac{1}{e^{\epsilon/2}+1}$ (line 10).

We give an example to show the difference between the encoded value of a high-frequent item and that of a low-frequent item in phase 2.

\begin{example}
  Suppose there is a high-frequent item $d_{high}$ and a low-frequent item $d_{low}$. The encoded value of $d_{high}$ is undoubtedly [-1,-1,-1,-1,-1,-1].
  Suppose the length of the sketch is $m=6$ and the hash value of $d_{low}$ is $h_{j}(d_{low})=2$, thus the encoded value of $d_{low}$ is [-1,-1,+1,-1,-1,-1].
\end{example}

We can learn from this example that the high-frequent items have no hash collisions with the low-frequent items in this way, since the high frequent items do not actually use the hash functions for encoding. The difference between the encoded value of a high-frequent item and a low-frequent item is at most 1 bit in our algorithm. Such a difference is smaller than that of two encoded values in the Apple-CMS. Since the Apple-CMS satisfies $\epsilon$-LDP, intuitively, our algorithm also satisfies $\epsilon$-LDP. We then formally prove that the client-side algorithm of LDPLCM satisfies $\epsilon$-LDP.

\begin{algorithm}
\caption{LDPLCMclient}\label{Algorithm:LDPLCMclient}
\hspace*{0.02in} {\bf Input: data $d\in D$, $phase$, privacy budget $\epsilon$, sketch parameters ($m$, $k$), frequency model $g$, threshold $P$}\\
\hspace*{0.02in} {\bf Output: perturbed value $\tilde{v}$, hash index $j$}
\begin{algorithmic}[1]
\If{$phase==1$} \Comment{Phase 1}
    \State $\tilde{v}, j\leftarrow$ $A_{client}\mbox{-}CMS$($d$, $\epsilon$, $m$, $k$)~\cite{APPLE}
\Else \Comment{Phase 2}
    \If  {$g(d)<P$} \Comment{$d$ is a low-frequent value}
        \State $\tilde{v}, j\leftarrow$ $A_{client}\mbox{-}CMS$($d$, $\epsilon$, $m$, $k$)~\cite{APPLE}
    \Else \Comment{$d$ is a high-frequent value}
        \State Sample $j$ uniformly at random from $[k]$
        \State Initialize a vector $v\leftarrow \{-1\}^m$
        \For {$i$ from 1 to $m$} \Comment{Perturb $v$}
            \State $\tilde{v}[i]\leftarrow v[i]\cdot(-1)$ with probability $\frac{1}{e^{\epsilon/2}+1}$
            \State $\tilde{v}[i]\leftarrow v[i]\cdot(+1)$ with probability $\frac{e^{\epsilon/2}}{e^{\epsilon/2}+1}$
        \EndFor
    \EndIf
\EndIf
\State \textbf{return: $\tilde{v}$, $j$}
\end{algorithmic}
\end{algorithm}

\begin{theorem}
  The algorithm LDPLCMclient satisfies $\epsilon$-LDP.
\end{theorem}

\begin{proof}
  Since the algorithm $A_{client}\mbox{-}CMS$ in reference~\cite{APPLE} satisfies $\epsilon$-LDP, we just need to prove that the phase 2 of Algorithm~\ref{Algorithm:LDPLCMclient} satisfies $\epsilon$-LDP. As the input of all the high-frequent items are the same, we only need to prove that the perturbed results of a high-frequent item $d$ and a low-frequent item $d'$ are sufficiently similar.
  We suppose the encode of $d$ and $d'$ are $v$ and $v'$, respectively. According to the algorithm, each bit of $v$ is -1, that is, the only difference between $v$ and $v'$ is on the $h_j(d')$-th bit, i.e., $v[h_j(d')]=-1$ and $v'[h_j(d')]=1$. We denote Algorithm~\ref{Algorithm:LDPLCMclient} by $A$ in the following proof.
  \begin{equation}\label{Equation:LDP-Proof1}
  \frac{Pr[A(d)=(\tilde{v},j)]}{Pr[A(d')=(\tilde{v},j)]}=\frac{Pr[perturb(v[h_j(d')])=\tilde{v}[h_j(d')]]}{Pr[perturb(v'[h_j(d')])=\tilde{v}[h_j(d')]]}
  =\frac{Pr[perturb(-1)=\tilde{v}[h_j(d')]]}{Pr[perturb(1)=\tilde{v}[h_j(d')]]}
  \end{equation}
  Since the algorithm perturbs $v[i]$ by multiplying (-1) and $v[i]$ with a probability  $\frac{1}{e^{\epsilon/2}+1}$,
\begin{equation}\label{Equantion:LDP-Proof2}
  e^{-\epsilon} < e^{-\epsilon/2}\le\frac{Pr[perturb(-1)=\tilde{v}[h_j(d')]]}{Pr[perturb(1)=\tilde{v}[h_j(d')]]}\le e^{\epsilon/2} < e^{\epsilon}
\end{equation}
Thus, the Algorithm~\ref{Algorithm:LDPLCMclient} satisfies $\epsilon$-LDP.
\end{proof}

\subsection{Server-side algorithm of LDPLCM}

The server-side of LDPLCM also has two phases. The main point of phase 1 is to build a frequency model based on the aggregations of the values from the sample clients. Since the frequency estimation based on the sampling are highly accurate for the high-frequent items, we use the frequency model to replace the storage of these high-frequent items. We leave the task of approximately storing and estimating the frequencies of low-frequent items to the sketch in phase 2.

The pseudo-code of the server-side algorithm is shown in Algorithm~\ref{Algorithm:LDPLCMserver-construction}. In phase 1, the server constructs the sketch with the perturbed inputs from the sample clients (line 2). Then, the algorithm randomly chooses $t$ sample items from the data domain (line 3). We use the sketch $M$ to estimate the frequency of each item $d_i$ according to the server-side algorithm of Apple-CMS (line 5). Since the server in phase 1 only gets information from the sample clients, the frequency of each item in each client should be scaled with $\frac{1}{r}$, where $r$ is the sampling rate of phase 1. The algorithm then trains the frequency model $g$ by mapping each $d_i$ to its estimated frequency $\hat{f}(d_i)$ (line 7). We calculate the predictions of each item $d_i$ and sort the predictions in descending order (line 8). We compute the frequency boundary $P$ separating the high-frequent items and low-frequent items by accumulating the frequencies of high-frequent items until it reaches $\theta$ times the sum of all the data (line 9-10). In phase 2, the algorithm constructs the sketch $M$ with the perturbed inputs of the clients other than the samples (line 12).

\begin{algorithm}
\caption{LDPLCMserver}\label{Algorithm:LDPLCMserver-construction}
\hspace*{0.02in} {\bf Input: data $d$, privacy budget $\epsilon$, sketch parameters ($m$, $k$), ratio of high-frequent items $\theta$, sampling rate $r$}\\
\hspace*{0.02in} {\bf Output: sketch $M$, frequency boundary $P$, frequency model $g$}
\begin{algorithmic}[1]
\If {$phase==1$} \Comment{Phase 1}
    \State Construct sketch $M$ with inputs from clients.
    %\State $P,g \leftarrow FrequencyModel(M, \theta)$
    \State $\{d_1, d_2,..., d_t\}\leftarrow$ samples randomly chosen from the data domain.
    \For {$i$ from 1 to $t$}
        \State $\hat{f}(d_i)\leftarrow A_{server}\mbox{-}CMS(d_i)/r$
    \EndFor
    \State Train frequency model $g: d_i\rightarrow \hat{f}(d_i)$. \Comment{Frequency Model $g$}
    \State $SortedFre\leftarrow Sort_{i\in [1,t]}\{g(d_i)\}$
    \State $id\leftarrow \max\{p| \sum_{i=1}^{p}SortedFre[i]\le \theta\cdot\sum_{i=1}^{t}SortedFre[i]\}$ \Comment{Frequency Boundary $P$}
    \State $P\leftarrow SortedFre[id]$
\ElsIf {$phase==2$} \Comment{Phase 2}
    \State Construct sketch $M$ with inputs from clients.
\EndIf
\end{algorithmic}
\end{algorithm}

%\begin{algorithm}
%\caption{LDPLCMserver}\label{Algorithm:LDPLCMserver-construction}
%\hspace*{0.02in} {\bf Input:data $d$, $phase$, privacy budget $\epsilon$, sketch parameters ($m$, $k$), ratio of high-frequent items $\theta$, sampling rate $r$}\\
%\hspace*{0.02in} {\bf Output: sketch $M$, Frequency Boundary $P$, Frequency Model $g$}
%\begin{algorithmic}[1]
%\If {$phase==1$} \Comment{Phase 1}
%    \State Construct sketch $M$ with inputs from clients.
%    %\State $P,g \leftarrow FrequencyModel(M, \theta)$
%    \State $\{d_1, d_2,..., d_t\}\leftarrow$ samples randomly chosen from the data domain.
%    \For {$i$ from 1 to $t$}
%        \State $\hat{f}(d_i)\leftarrow A_{server}\mbox{-}CMS(d_i)/r$
%    \EndFor
%    \State $SortedFre\leftarrow Sort_{i\in [1,t]}\{\hat{f}(d_i)\}$
%    \State $P\leftarrow Max\{SortedFre[p]|\theta \sum_{i=1}^{p}SortedFre[i]\le \sum_{i=1}^{n}SortedFre[i]\}$ \Comment{Frequency Boundary $P$}
%    \State Train frequency model $g: d_i\rightarrow \hat{f}(d_i)$. \Comment{Frequency Model $g$}
%\ElsIf {$phase==2$} \Comment{Phase 2}
%    \State Construct sketch $M$ with inputs from clients.
%\EndIf
%\end{algorithmic}
%\end{algorithm}

After these two phases, we get a model predicting the frequencies for the high-frequent items and a sketch approximately estimating the frequencies of the low-frequent items. The LDPLCM estimates the frequency of high-frequent items and low-frequent items in different ways. The high-frequent items are predicted according to the model, and the low-frequent items are estimated based on the sketch. It is worth noting that we need to eliminate the error caused by high-frequent items from the sketch-based estimation since the LDPLCM involves some dummy values of high-frequent items into the sketch to avoid privacy leaks. As each bit of the dummy values is flipped with the same probability and the line to insert a perturbed dummy value is chosen randomly, the perturbed dummy values almost identically influent each cell of the sketch. In addition, we can evaluate the impact of involving all these dummy values of the high-frequent items on each cell of the sketch, since the total frequencies of high-frequent items can be computed by $\theta\cdot n$, where $\theta$ is the ratio of the total frequencies of high-frequent items to the total frequencies of all the items. At such, we can accurately evaluate these errors and eliminate them from the estimation.

The pseudo-code of our estimator is shown in Algorithm~\ref{Algorithm:LDPLCM-estimate}. If the model-based prediction $g(d)$ exceeds the boundary $P$, then the value $d$ is judged to be a high-frequent value. Thus, its estimated frequency is the prediction $g(d)$ (line 2). Otherwise, the algorithm estimates the frequency of a low-frequent value according to the sketch $M$ (line 4).

\begin{algorithm}
\caption{LDPLCM-estimator}\label{Algorithm:LDPLCM-estimate}
\hspace*{0.02in} {\bf Input: data $d$, sketch parameters ($m$, $k$), frequency model $g$, boundary $P$, ratio of high-frequent items $\theta$}\\
\hspace*{0.02in} {\bf Output: frequency estimation $\hat{f}(d)$}
\begin{algorithmic}[1]
    \If {$g(d)>P$}
        \State $\hat{f}(d)=g(d)$
    \Else
         \State $\hat{f}(d)=\frac{m}{m-1}(\frac{1}{k}\sum_{l=1}^{k}M_{l,h_l(d)}-(1-\theta)\frac{n}{m})$
    \EndIf
\State \textbf{return: $\hat{f}(d)$}
\end{algorithmic}
\end{algorithm}

The following theorem proves that the output of our LDPLCM-estimator is an unbiased estimate of the frequency for a low-frequent item.

\begin{theorem}
  The estimated frequency of a low-frequent item provided by the LDPLCM is a unbiased estimate of $f(d)$, i.e., $\mathbb{E}[\frac{m}{m-1}(\frac{1}{k}\sum_{l=1}^{k}M_{l,h_l(d)}-(1-\theta)\frac{n}{m})]=f(d)$.
\end{theorem}

\begin{proof}
The LDPLCM estimates the frequency of a low-frequent item $d$ as $\hat{f}(d)$. We first analyze the contribution of each data entry $d^{(i)}$ to the estimation $\hat{f}(d)$.
  Let $Perturb(v^{(i)})=(-1)\cdot v^{(i)}$ with probability $\frac{e^{\epsilon/2}}{1+e^{\epsilon/2}}$ and $J\sim Uniform[k]$. The encoding vector of a high-frequent item is $v^{(i)}\in\{-1\}^m$, and the encoding vector of a low-frequent item is $-1$ everywhere except at position $h_j(d^{(i)})$ for record $i$.
  We use $M(j,h_j(d))^{(i)}$ to denote the contribution of the $i$th data to the $j$ line and $h_j(d)$ column of the sketch $M$.
  \begin{equation}\label{Equation:UnbiasedEst}
    M(j,h_j(d))^{(i)}=k(\frac{c_\epsilon \mathrm{Perturb}( v^{(i)}[h_j(d)])+1}{2})\mathbbm{1}\{J=j\},
  \end{equation}
  where $c_\epsilon=\frac{e^{\epsilon/2}+1}{e^{\epsilon/2}-1}$, and  $M(j,h_j(d))^{(i)}$ is nonzero only when $J=j$.

  Then, we analyze the expectation of $\mathbb{E}[ M(j,h_j(d))^{(i)}]$ under different conditions of the $i$th entry $d^{(i)}$.

  (1) If $d^{(i)}$ is a high-frequent item,
  \begin{align}\label{Equation:UnbiasedCondition1}
    \mathbb{E}[ M(j,h_j(d))^{(i)}] &= k(\frac{c_\epsilon \mathrm{Perturb}( v^{(i)}[h_j(d)])+1}{2})\mathrm{Pr}\{J=j\}\\
    &=\mathbb{E}[\frac{c_\epsilon \mathrm{Perturb}(-1)+1}{2}]=0
  \end{align}
  (2) $d^{(i)}$ is a low-frequent item, and $d^{(i)}=d$,
  \begin{align}\label{Equation:UnbiasedCondition2}
    \mathbb{E}[ M(j,h_j(d))^{(i)}] &= k(\frac{c_\epsilon \mathrm{Perturb}( v^{(i)}[h_j(d)])+1}{2})\mathrm{Pr}\{J=j\}=1\\
  \end{align}
  (3) $d^{(i)}$ is a low-frequent item, and $d^{(i)}\neq d$
  \begin{align}\label{Equation:UnbiasedCondition3}
    \mathbb{E}[ M(j,h_j(d))^{(i)}] &= k(\frac{c_\epsilon \mathrm{Perturb}( v^{(i)}[h_j(d)])+1}{2})\mathrm{Pr}\{J=j\}\\
    &=(1-\frac{1}{m})\mathbb{E}[\frac{c_\epsilon \mathrm{Perturb}(-1)+1}{2}]+\frac{1}{m}\mathbb{E}[\frac{c_\epsilon \mathrm{Perturb}(+1)+1}{2}]\\
    &=(1-\frac{1}{m})\cdot 0+\frac{1}{m}\cdot 1=\frac{1}{m}
  \end{align}

  Thus, $\mathbb{E}[ M(j,h_j(d))^{(i)}]=0\cdot \mathbbm{1}\{d_{high}^{(i)} \} + 1\cdot \mathbbm{1}\{d_{low}^{(i)}=d \}+ \frac{1}{m}\cdot \mathbbm{1}\{d_{low}^{(i)}\neq d\}$, where $d_{high}^{(i)}$ means a high-frequent item, and $d_{low}^{(i)}$ means a low-frequent item.
\begin{align}\label{Equation}
  \mathbb{E}[\frac{1}{k}\sum_{i=1}^{n}\sum_{j=1}^{k}M(j,h_j(d))^{(i)}]&=
  0\cdot n\cdot \theta + 1\cdot f(d) + \frac{1}{m}\cdot [n(1-\theta)-f(d)] \\
  &=(1-\frac{1}{m})f(d)+(1-\theta)\frac{n}{m}
\end{align}
The expectation of the LDPLCM-based estimation:

\begin{equation}\label{Equation:UnbiasedResult}
  \mathbb{E}[\hat{f}(d)]=\mathbb{E}[\frac{m}{m-1}(\frac{1}{k}\sum_{i=1}^{n}\sum_{j=1}^{k}M(j,h_j(d))^{(i)}-(1-\theta)\frac{n}{m})]=f(d)
\end{equation}

Thus, the output of LDPLCM $\hat{f}(d)$ is a unbiased estimate of $f(d)$.
\end{proof}

To ensure the utility of LDPLCM-based frequency estimation, we prove that the variance of the estimation is limited. Before computing the variance, we first prove some corresponding lemmas.

Lemma~\ref{Lemma:ExpectationSquaredX} computes the expectation of the squared entry of the $i$th data to the $j$ line and $h_j(d)$ column of the sketch $M$. This lemma will be used to compute the variance of LDPLCM-based estimation.

\begin{lemma}\label{Lemma:ExpectationSquaredX}
$\mathbb{E}[( M(j,h_j(d))^{(i)})^2]=\frac{k}{4}(c^2_\epsilon-1)+k\cdot\mathbbm{1}\{d_{low}^{(i)}=d \}+ \frac{k}{m}\cdot \mathbbm{1}\{d_{low}^{(i)}\neq d\}$
\end{lemma}

\begin{proof}
  (1) If $d^{(i)}$ is a high-frequent item,
  \begin{align}\label{Equation:LemmaCondition1}
    \mathbb{E}[ (M(j,h_j(d))^{(i)})^2] &= k^2\mathbb{E}[ (\frac{c_\epsilon \mathrm{Perturb}( v^{(i)}[h_j(d)])+1}{2})^2\mathbbm{1}\{J=j\}]\\
    &=k\mathbb{E}[(\frac{c_\epsilon \mathrm{Perturb}(-1)+1}{2})^2]\\
    &=k((\frac{-c_\epsilon+1}{2})^2(\frac{e^{\epsilon/2}}{1+e^{\epsilon/2}})+(\frac{c_\epsilon+1}{2})(\frac{1}{1+e^{\epsilon/2}}))\\
    &=\frac{k}{4}({c_\epsilon}^2-1)
  \end{align}
  (2) $d^{(i)}$ is a low-frequent item, and $d^{(i)}=d$,
  \begin{align}\label{Equation:LemmaCondition2}
    \mathbb{E}[ (M(j,h_j(d))^{(i)})^2] &= k^2\mathbb{E}[ (\frac{c_\epsilon \mathrm{Perturb}( v^{(i)}[h_j(d)])+1}{2})^2\mathbbm{1}\{J=j\}]\\
    &=k\mathbb{E}[(\frac{c_\epsilon \mathrm{Perturb}(+1)+1}{2})^2]\\
    &=k((\frac{c_\epsilon+1}{2})^2(\frac{e^{\epsilon/2}}{1+e^{\epsilon/2}})+(\frac{-c_\epsilon+1}{2})(\frac{1}{1+e^{\epsilon/2}}))\\
    &=\frac{k}{4}({c_\epsilon}^2-1)+k
  \end{align}
  (3) $d^{(i)}$ is a low-frequent item, and $d^{(i)}\neq d$
  \begin{align}\label{Equation:LemmaCondition3}
    \mathbb{E}[ (M(j,h_j(d))^{(i)})^2] &= k^2\mathbb{E}[ (\frac{c_\epsilon \mathrm{Perturb}( v^{(i)}[h_j(d)])+1}{2})^2\mathbbm{1}\{J=j\}]\\
    &=\frac{k}{m}\mathbb{E}[(\frac{c_\epsilon \mathrm{Perturb}(+1)+1}{2})^2]+k(1-\frac{1}{m})\mathbb{E}[(\frac{c_\epsilon \mathrm{Perturb}(-1)+1}{2})^2]\\
    &=\frac{k}{4}({c_\epsilon}^2-1)+\frac{k}{m}
  \end{align}

  Thus, $\mathbb{E}[( M(j,h_j(d))^{(i)})^2]=\frac{k}{4}(c^2_\epsilon-1)+0\cdot\mathbbm{1}\{d_{high}^{(i)}\} +k\cdot\mathbbm{1}\{d_{low}^{(i)}=d \}+ \frac{k}{m}\cdot \mathbbm{1}\{d_{low}^{(i)}\neq d\}$, where $d_{high}^{(i)}$ means a high-frequent item, and $d_{low}^{(i)}$ means a low-frequent item.
\end{proof}

Lemma~\ref{Lemma:Covariance} computes the covariance of the entries $d^{(i_1)}$, $d^{(i_2)}$ to the sketch $M$ in different cases.

\begin{lemma}\label{Lemma:Covariance}
Let $i_1\neq i_2$ be different indices. We have:

(1)If $j_1\neq j_2$, then
\begin{equation}
\mathrm{Cov}(M(j_1,h_{j_1}(d))^{(i_1)}, M(j_2,h_{j_2}(d))^{(i_2)})=0.
\end{equation}

(2)If $d^{(i_1)}$ and $d^{(i_2)}$ are low-frequent values, $d^{(i_1)}=d$ or $d^{(i_2)}=d$ or $d^{(i_1)}\neq d^{(i_2)}$, then
\begin{equation}
\mathrm{Cov}(M(j,h_j(d))^{(i_1)}, M(j,h_j(d))^{(i_2)})=0.
\end{equation}

(3)If $d^{(i_1)}$ and $d^{(i_2)}$ are low-frequent values, $d^{(i_1)}, d^{(i_2)}=d^*$ and $d^*\neq d$, then
\begin{equation}
\mathrm{Cov}(M(j,h_j(d))^{(i_1)}, M(j,h_j(d))^{(i_2)})=\frac{1}{m}-\frac{1}{m^2}.
\end{equation}

(4)If $d^{(i_1)}$ or $d^{(i_2)}$ is a high-frequent value, then
\begin{equation}
\mathrm{Cov}(M(j,h_j(d))^{(i_1)}, M(j,h_j(d))^{(i_2)})=0.
\end{equation}
\end{lemma}

\begin{proof}
The (1),(2), and (3) are proved in Reference~\cite{APPLE}. We just need to prove the (4).

\begin{align}
  &\mathrm{Cov}(M(j,h_j(d))^{(i_1)}, M(j,h_j(d))^{(i_2)})\\
  &=\mathbb{E}[M(j,h_j(d))^{(i_1)}\cdot M(j,h_j(d))^{(i_2)}]-
  \mathbb{E}[M(j,h_j(d))^{(i_1)}]\mathbb{E}[M(j,h_j(d))^{(i_2)}]\\
  &=\frac{1}{4}\{\mathbb{E}[v^{(i_1)}[h_j(d)]\cdot v^{(i_2)}[h_j(d)]]+\mathbb{E}[v^{(i_1)}[h_j(d)]]+\mathbb{E}[v^{(i_2)}[h_j(d)]]+1\}\nonumber\\
  & -\mathbb{E}[M(j,h_j(d))^{(i_1)}]\cdot \mathbb{E}[M(j,h_j(d))^{(i_2)}]
\end{align}
If $d^{(i_1)}$ is a high-frequent value, then $v^{(i_1)}[h_j(d)]=-1$ and $\mathbb{E}[M(j,h_j(d))^{(i_1)}]=0$.

Thus, $\mathrm{Cov}(M(j,h_j(d))^{(i_1)}, M(j,h_j(d))^{(i_2)})=0$.
\end{proof}

Theorem~\ref{Theorem:Variance} uses the above Lemma~\ref{Lemma:ExpectationSquaredX} and Lemma~\ref{Lemma:Covariance} to compute the variance of the LDPLCM-based frequency estimation for the low-frequent items.

\begin{theorem}\label{Theorem:Variance}
    The variance of the LDPLCM-based frequency estimation for a low-frequent item is limited.
    That is $\mathrm{Var}[\frac{1}{k}\sum_{i=1}^{n}\sum_{j=1}^{k}M(j,h_j(d))^{(i)}]
    <\frac{n(c_\epsilon^2-1)}{4}+\frac{n(1-\theta)-f(d)}{m}
    +\frac{1}{km}\sum_{d^*\in D_{low}}f(d^*)^2$.
\end{theorem}

\begin{proof}
We compute the variance of the LDPLCM-based estimation for a low-frequent item as follows.
  \begin{align}
    &\mathrm{Var}[\frac{1}{k}\sum_{i=1}^{n}\sum_{j=1}^{k}M(j,h_j(d))^{(i)}]\\
    &=\sum_{i=1}^{n}Var[\frac{1}{k}\sum_{j=1}^{k}M(j,h_j(d))^{(i)}]+\frac{1}{k^2}\sum_{i_1\neq i_2}\mathrm{Cov}(\sum_{j=1}^{k}M(j,h_j(d))^{(i_1)},\sum_{j'=1}^{k}M(j,h_j(d))^{(i_2)})\\
    &=\frac{1}{k^2}\sum_{i=1}^{n}\sum_{j=1}^{k}\mathrm{Var}[M(j,h_j(d))^{(i)}]+\frac{1}{k^2}\sum_{i_1\neq i_2}\sum_{j=1}^{k}\mathrm{Cov}(M(j,h_j(d))^{(i_1)}, M(j,h_j(d))^{(i_2)})
   \end{align}
According to Lemma~\ref{Lemma:Covariance},
   \begin{align}
    &\frac{1}{k^2}\sum_{i_1\neq i_2}\sum_{j=1}^{k}\mathrm{Cov}(M(j,h_j(d))^{(i_1)}, M(j,h_j(d))^{(i_2)})\\
    &=(\frac{1}{km}-\frac{1}{km^2})\sum_{d^*\neq d}\sum_{d_{low}^{(i_1)}=d^*}\sum_{i_1\neq i_2}\mathbbm{1}\{d_{low}^{(i_2)}=d^*\}\\
    &=(\frac{1}{km}-\frac{1}{km^2})\sum_{d^*\neq d}\sum_{d_{low}^{(i)}=d^*}(f(d^*)-1)
   \end{align}
   where $d_{low}^{(i_1)}$ and $d_{low}^{(i_2)}$ are two low-frequent items, and $d_{low}^{(i_1)}=d_{low}^{(i_2)}=d^*$ and $d^*\neq d$.

Thus,
   \begin{align}
    &\mathrm{Var}[\frac{1}{k}\sum_{i=1}^{n}\sum_{j=1}^{k}M(j,h_j(d))^{(i)}]\\
    &=\frac{1}{k^2}\sum_{i=1}^{n}\sum_{j=1}^{k}(\mathbb{E}[(M(j,h_j(d))^{(i)})^2]-(\mathbb{E}(M(j,h_j(d))^{(i)}])^2)+
    (\frac{1}{km}-\frac{1}{km^2})\sum_{d^*\neq d}\sum_{d_{low}^{(i)}=d^*}(f(d^*)-1)\\
    &=\frac{n(c_\epsilon^2-1)}{4}+\frac{n(1-\theta)-f(d)}{m}(1-\frac{1}{m})+
    (\frac{1}{km}-\frac{1}{km^2})\sum_{d^*\neq d}\sum_{d_{low}^{(i)}=d^*}(f(d^*)-1)\\
    &=\frac{n(c_\epsilon^2-1)}{4}+\frac{n(1-\theta)-f(d)}{m}(1-\frac{1}{m}-\frac{1}{k}+\frac{1}{km})
    +(\frac{1}{km}-\frac{1}{km^2})\sum_{d^*\neq d}f(d^*)^2
    \end{align}

Since $d^*$ is the value of a low-frequent item,  $\sum_{d^*\neq d}f(d^*)^2=\sum_{d^*\in D_{low}}f(d^*)^2-f(d)^2$. The variance of LDPLCM is limited by the following equation.
\begin{equation}
  \mathrm{Var}[\frac{1}{k}\sum_{i=1}^{n}\sum_{j=1}^{k}M(j,h_j(d))^{(i)}]<\frac{n(c_\epsilon^2-1)}{4}+\frac{n(1-\theta)-f(d)}{m}+\frac{1}{km}\sum_{d^*\in D_{low}}f(d^*)^2
\end{equation}
 %Since $1-\theta<1$ and $\sum_{d_{low}^*\neq d}f(d_{low}^*)^2<\sum_{d^*\neq d}f(d^*)^2$, the variance of the estimation provided by LDPLCM for a low-frequent item is smaller than that provided by the apple-CMS.
\end{proof}

The variance of Apple-CMS is $\frac{n(c_\epsilon^2-1)}{4}+\frac{n-f(d)}{m}(1-\frac{1}{m}-\frac{1}{k}+\frac{1}{km})
    +(\frac{1}{km}-\frac{1}{km^2})(\sum_{d^*\in D}f(d^*)^2-f(d)^2)$.
According to the Theorem~\ref{Theorem:Variance}, our LDPLCM is more accurate for the low-frequent items compared with the Apple-CMS.

% Experimental section

\section{Experimental Result}

\subsection{Settings}

\noindent \underline{Hardware and Library}

All the experiments are implemented in Python 3.8 and run on Intel(R) Xeon(R) Gold 5218 CPU, with 1TB of MEMORY and 256GB of RAM.

\noindent \underline{Datasets}

We test the performance of frequency estimation algorithms on three datasets, including two synthetic datasets and one real-world dataset.

(1)The first dataset is a synthetic dataset including 10 million records generated from a Zipf distribution with skewness $s=1.1$. The domain size is 2,817,991. We refer to this dataset as Zipf(10M) in the following experiments.

(2)The second dataset is a synthetic dataset including 100 million records generated from a Zipf distribution with skewness $s=1.1$. The domain size is 22,774,443. We refer to this dataset as Zipf(100M) in the following experiments.

(3)The third one is a real-world dataset named Wesad~\cite{schmidt2018introducing} for wearable stress and affect detection. This dataset is a 16GB dataset containing 63 million records. We use the attribute RESPIRATION for the frequency estimation experiments. The domain size is 44,900. We refer to this dataset as WESAD in the following experiments.

\noindent \underline{Error Metrics}

We use metrics to measure the accuracy of different algorithms. In the following metrics, $d$ denotes the size of data domain, $i$ denotes the $i$-th data item, $f(x_{i})$ and $\hat{f}(x_{i})$ represent the true value and the predicted value, respectively.

(1) Mean squared error (MSE): $\frac{1}{d} \sum_{x_{i}\in D} (f(x_{i})- \hat{f}(x_{i}))^{2}$.

(2) Sum squared error (SSE): $ \sum_{1}^{i} (f(x_{i})- \hat{f}(x_{i}))^{2}$.

\noindent \underline{Parameters}

$\theta$: the ratio of the total frequencies of high-frequent items to the frequencies of all items, $\theta=\frac{\sum_{x_{i}\in HighFre}f(x_{i})}{\sum_{x\in D}f(x_{i})}$, where $HighFre$ means the set of all the high-frequent items.

$r$: the sampling rate for the LDPLCM algorithm.

$\epsilon$: the privacy budget of clients in LDPLCM, Apple-CMS, and Apple-HCMS algorithms.

$(m, k)$: the sketch matrix parameters with hash functions $k$ and the domain size $m$ in LDPLCM, Apple-CMS, and Apple-HCMS algorithms.

$k'$: the parameter is to restrict clients to uniformly choosing from $k'$ hash functions in FLH.

\noindent \underline{Frequency Model}

The frequency model should be both lightweight and accurate. On the one hand, it must be small enough to be easily passed between the server and the clients. On the other hand, it should be sufficiently accurate to predict the frequency for the high-frequent items. For the sake of fairness, we compared the accuracy of the Random Forest Regressor model and the Gradient Boosting Regressor model with the same space cost on three datasets. Figure~\ref{Fig:model_compared} shows that the Gradient Boosting Regressor model is more accurate than the Random Forest Regressor model. Thus, we adopt the Gradient Boosting Regressor model to form the frequency model of LDPLCM. The advantage of the Gradient Boosting Regressor model is that it introduces a new weak classifier in each iteration to reduce the residuals of the previously existing classifier combinations, which results in improved fitting and prediction data rather than just using weak learning algorithms.

\begin{figure*}
\centering
    \includegraphics[scale=0.4]{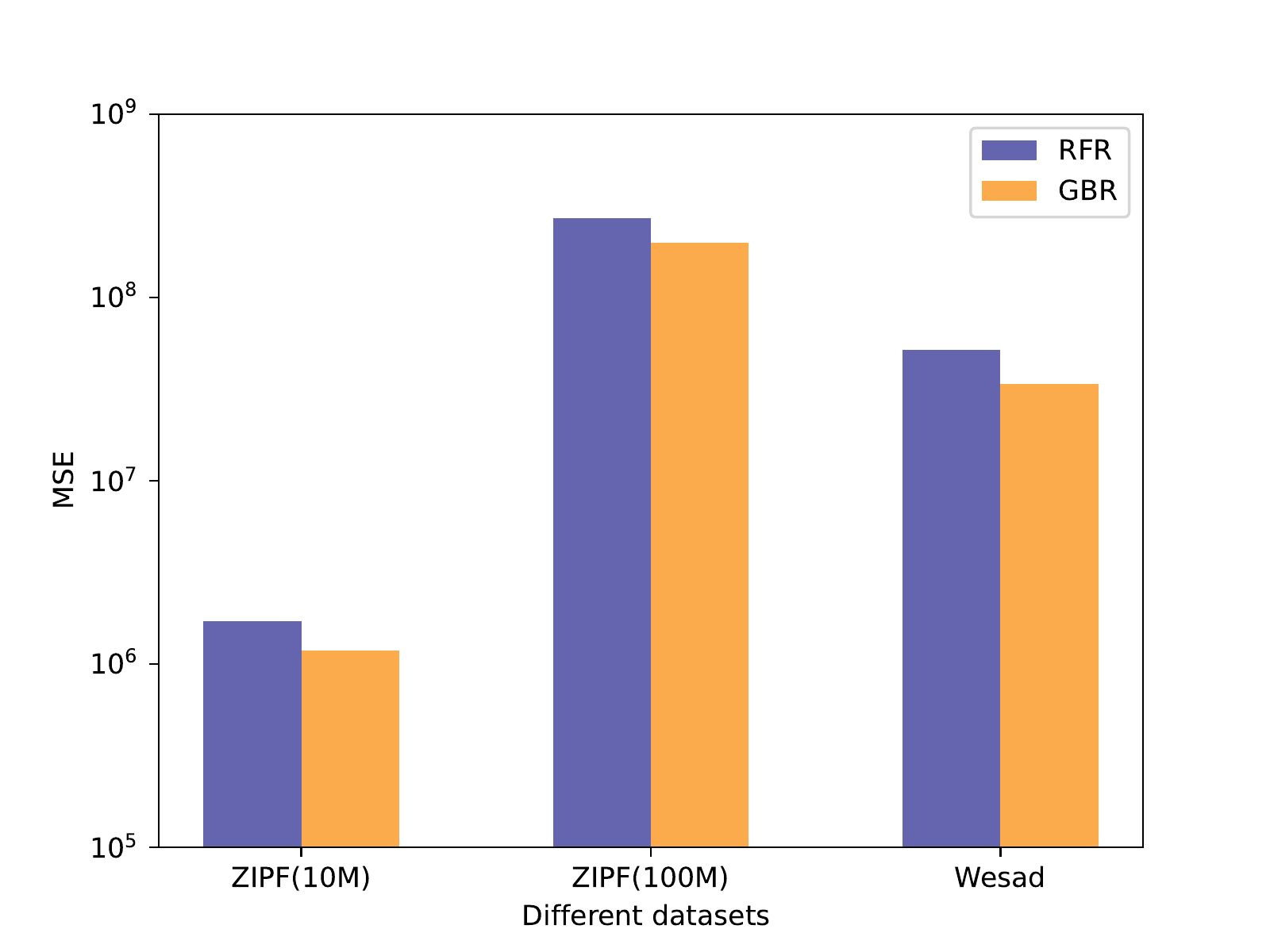}
\caption{The performance on different models for all datasets.}
\label{Fig:model_compared}
\end{figure*}

We set the parameters of the models as $learning rate=0.05,$ $estimators=350,$ $max depth=5$ for Zipf(10M), Zipf(100M) datasets, and $learning rate=0.1$, $estimators=100$, $max depth=3$ for WESAD dataset. These parameters are tuned according to the GridSearchCV\footnote{\url{https://scikit-learn.org/stable/modules/generated/sklearn.model_selection.GridSearchCV.html?highlight=gridsearchcv#sklearn.model_selection.GridSearchCV}}.

\noindent \underline{Competitors}

In the following experiments, we compare the performance of our LDPLCM with state-of-the-art algorithms including Apple-CMS and Apple-HCMS.

(1)Apple-CMS: The clients flip each bit of their one-hot vectors with the probability $\frac{1}{e^{\frac{\epsilon}{2}+1}}$ and send the perturbed items to the server. Then, the server stores the perturbed items in a $k \times m$ sketch matrix and estimates by averaging the hash entries.

(2)Apple-HCMS: The clients compute the Hadamard transform of any one-hot vectors and flip with the probability $\frac{1}{e^{\frac{\epsilon}{2}+1}}$. The algorithm in the server is as same as Apple-CMS.

(3)FLH~\cite{cormode2021frequency}: The clients are restricted by a parameter $k'$ to choose hash functions. The hash functions map the data domain $[d]$ to $[g]$, where $g=e^\epsilon$. The service calculates which domain elements a client's perturbed item contributes frequency towards by pre-computing a $k' \times d$ matrix.

\subsection{Accuracy}

\begin{figure*}
\centering
  \subfigure[Zipf(10M).]{
    \label{Fig:10zipf_accuracy}
    \includegraphics[width=0.3\linewidth]{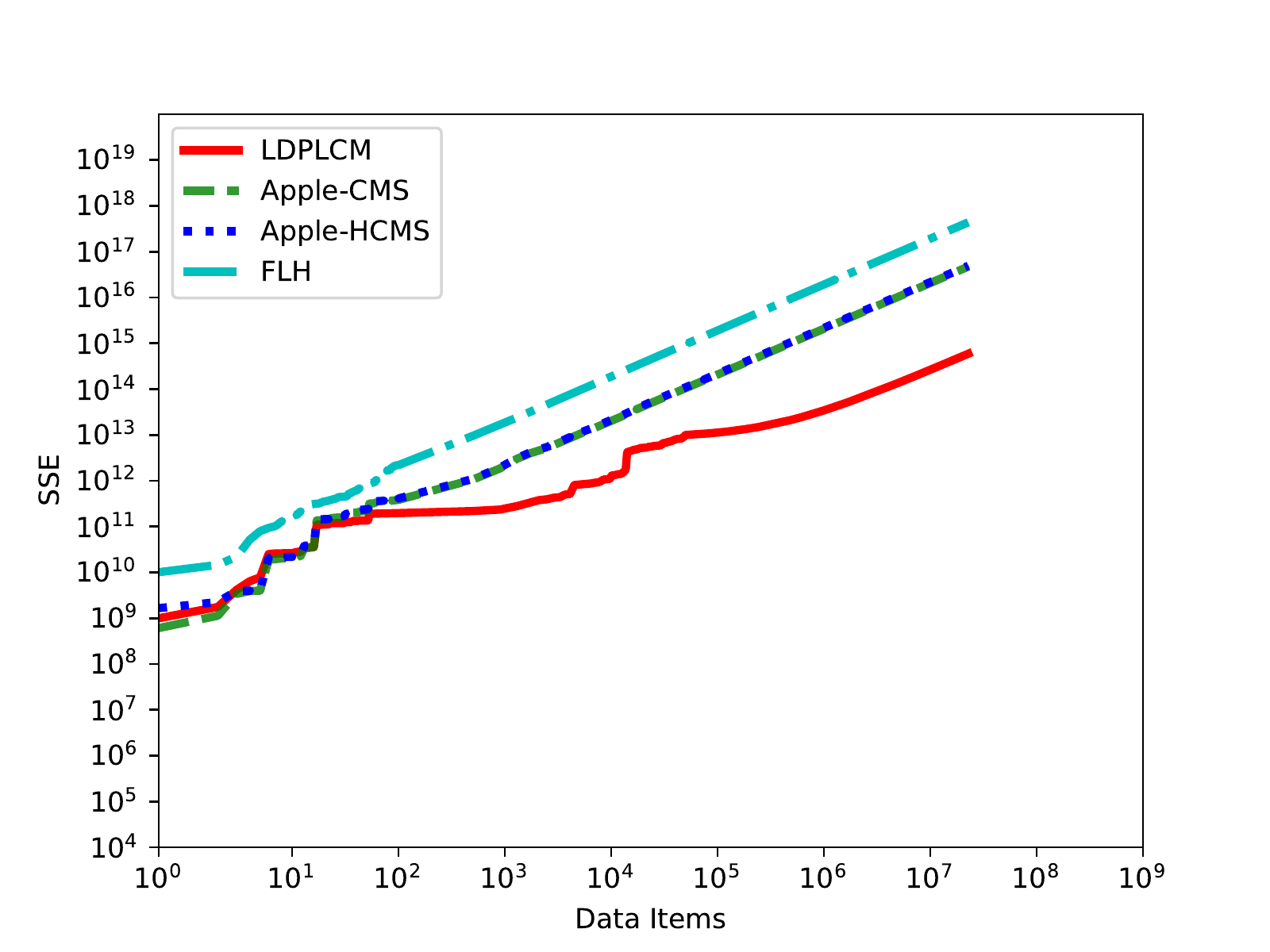}}
  \subfigure[Zipf(100M).]{
    \label{Fig:100zipf_accuracy}
    \includegraphics[width=0.3\linewidth]{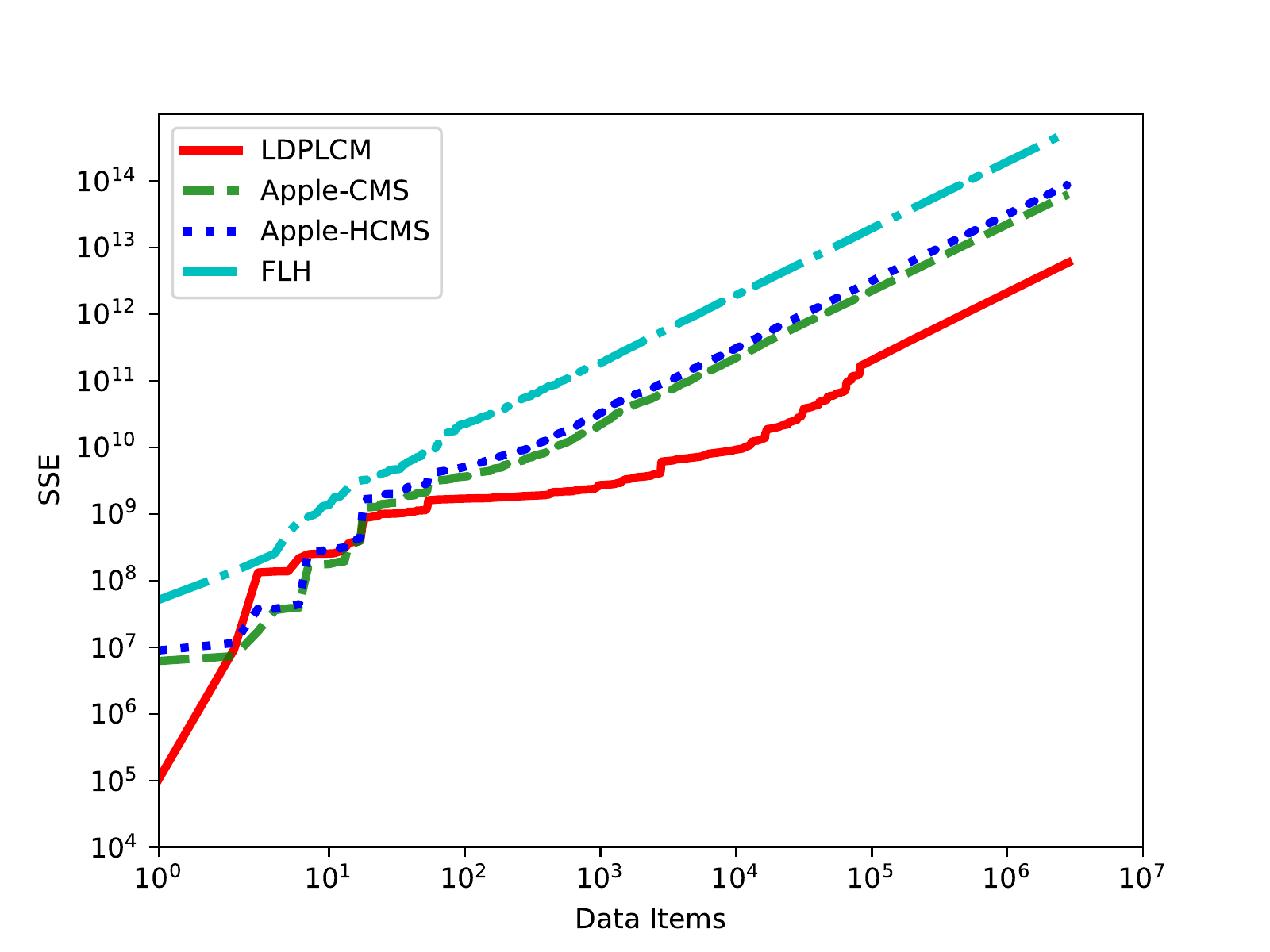}}
  \subfigure[Wesad.]{
    \label{Fig:Wesad_accuracy}
    \includegraphics[width=0.3\linewidth]{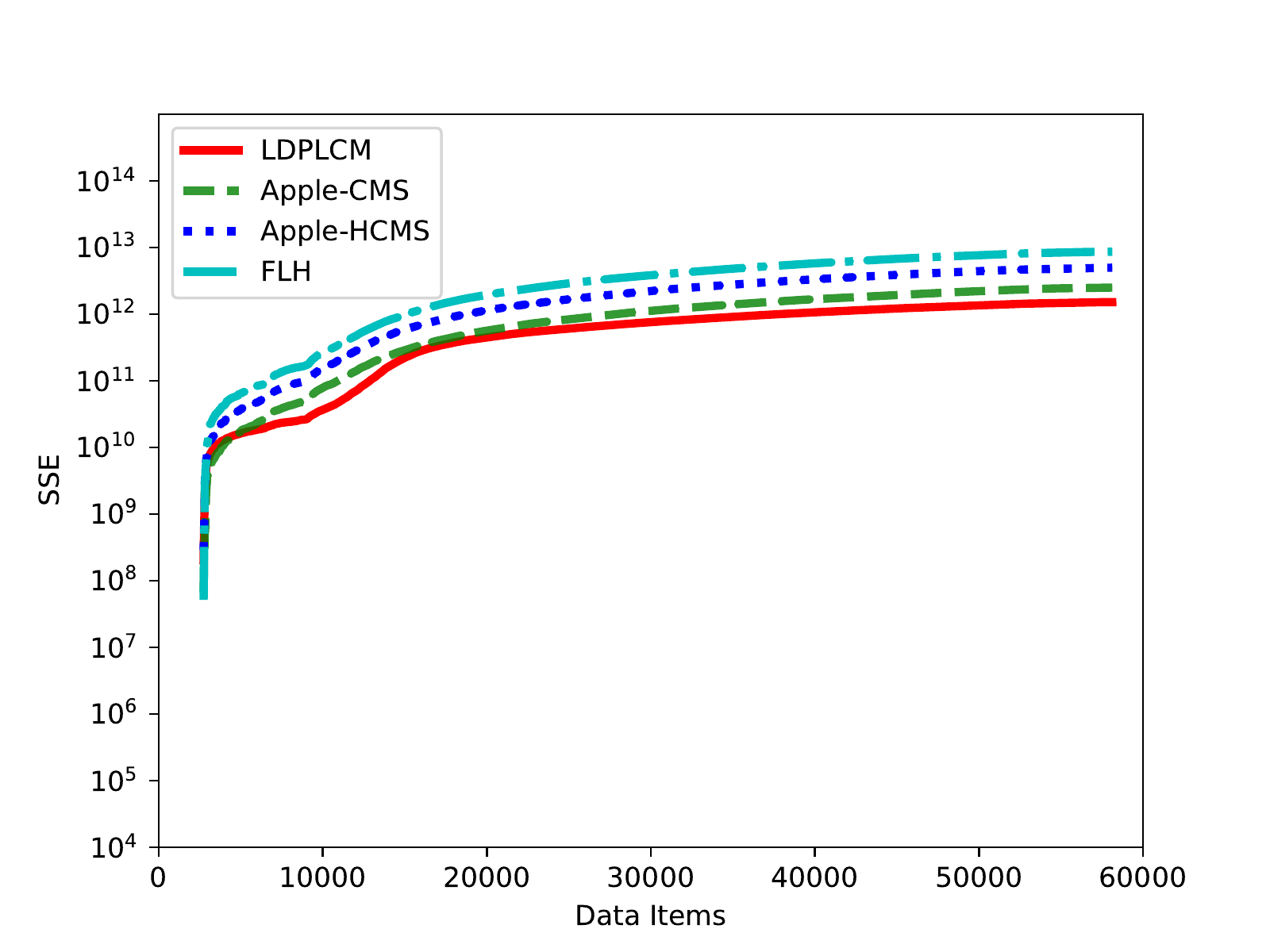}}
\caption{The Accuracy of Frequency Estimation}
\label{Fig:accuracy}
\end{figure*}

In this experiment, we compare the accuracy of our method with the Apple-CMS and Apple-HCMS on the three datasets. We set $r=0.1$ and $\epsilon=4$.
The experimental results on Zipf(10M), Zipf(100M) and WESAD datasets are shown in Figure~\ref{Fig:10zipf_accuracy}, Figure~\ref{Fig:100zipf_accuracy}, and Figure~\ref{Fig:Wesad_accuracy}, respectively. The error is measured by SSE. For the Zipf(10M) and Zipf(100M), we set the sketch parameters ($m$, $k$) as $m=1024$, $k=64$ and $\theta=0.5$. For the WESAD dataset, we set the parameters as $m=512$, $k=32$ and $\theta=0.4$. We also set the parameter $k'=128$ of FLH for all datasets. Since the domain of WESAD is smaller than that of zipf datasets, we use a smaller sketch to test the impact of hash collisions.

We can learn from these figures that LDPLCM is more accurate than the Apple-CMS, Apple-HCMS and FLH on different datasets. The reason is that we train a model on some samples to distinguish the high-frequent and low-frequent items and only use the sketch to estimate the low-frequent items. The high-frequent items in phase 2 are encoded as dummy values causing no collisions with the low-frequent items. The frequency estimation of FLH is worse than others obviously because it maps the data to a small domain $g=e^\epsilon+1$. And in this experiment, $g=56$, which is much smaller than the data domain.

We also can find that the LDPLCM estimation error of high-frequent items is slightly higher than the Apple-CMS and Apple-HCMS. It is reasonable since the high-frequent items are estimated based on the frequency model trained according to the samples, the estimations are slightly higher than the other methods that estimate frequency based on sketches built with the entire dataset. But the total errors of the estimation for all the items are lower than the other methods. It is because the estimations of our method for the low-frequent items are more accurate than the others. In addition, the sampling errors have little effect on the results for the larger datasets as shown in Figure~\ref{Fig:100zipf_accuracy}, which illustrates the applicability of our method to a larger domain.

\subsection{Space cost}

We test the accuracy of different methods with similar space costs on Zipf(100M). In this experiment, we only compare LDPLCM with Apple-CMS and Apple-HCMS since they use the same structure(sketch) to estimate the frequency. The space cost of LDPLCM includes the size of the frequency model and the size of the sketch. The space costs of Apple-CMS and Apple-HCMS are only the size of the sketches. We take a variety of settings to make different methods have similar space costs. The settings (the total space cost, the size of frequency model and the size of sketch) of different methods are shown in Table~\ref{Tab:space_cost}.

The experimental result about space cost and accuracy is shown in Figure~\ref{Fig:100zipf_spacecost}.
We can learn that the errors decrease with space cost. It is reasonable since a larger sketch can reduce hash collisions. It is clear that our LDPLCM is more accurate than the Apple-CMS and Apple-HCMS while occupying the similar space. The frequency model is more lightweight than sketch, so that it can be transferred more quickly between the clients and the server.

\begin{table}[hptb]
\centering
\caption{Settings for the space cost experiment.}
\label{Tab:space_cost}
\begin{tabular}{|lll|l|}
\hline
\multicolumn{3}{|l|}{LDPLCM}                                                                                                   & \multicolumn{1}{|l|}{Apple-CMS/Apple-HCMS}                   \\ \hline
\multicolumn{1}{|l|}{Total Space Cost(KB)} & \multicolumn{1}{l|}{Model Size(KB)} & Sketch Size(KB) & \multicolumn{1}{|l|}{Total Space Cost(KB)}\\ \hline
\multicolumn{1}{|l|}{685}                 & \multicolumn{1}{l|}{461}            & 224             & \multicolumn{1}{|l|}{896}                  \\ \hline
\multicolumn{1}{|l|}{2235}                 & \multicolumn{1}{l|}{1339}           & 896             & \multicolumn{1}{|l|}{1792}                 \\ \hline
\multicolumn{1}{|l|}{3141}                 & \multicolumn{1}{l|}{1349}           & 1792            & \multicolumn{1}{|l|}{3584}                 \\ \hline
\multicolumn{1}{|l|}{4921}                 & \multicolumn{1}{l|}{1337}           & 3584            & \multicolumn{1}{|l|}{7168}                 \\ \hline
\multicolumn{1}{|l|}{8514}                 & \multicolumn{1}{l|}{1346}           & 7168            & \multicolumn{1}{|l|}{8848}                 \\ \hline
\end{tabular}
\end{table}

\begin{figure*}
\centering
    \includegraphics[scale=0.4]{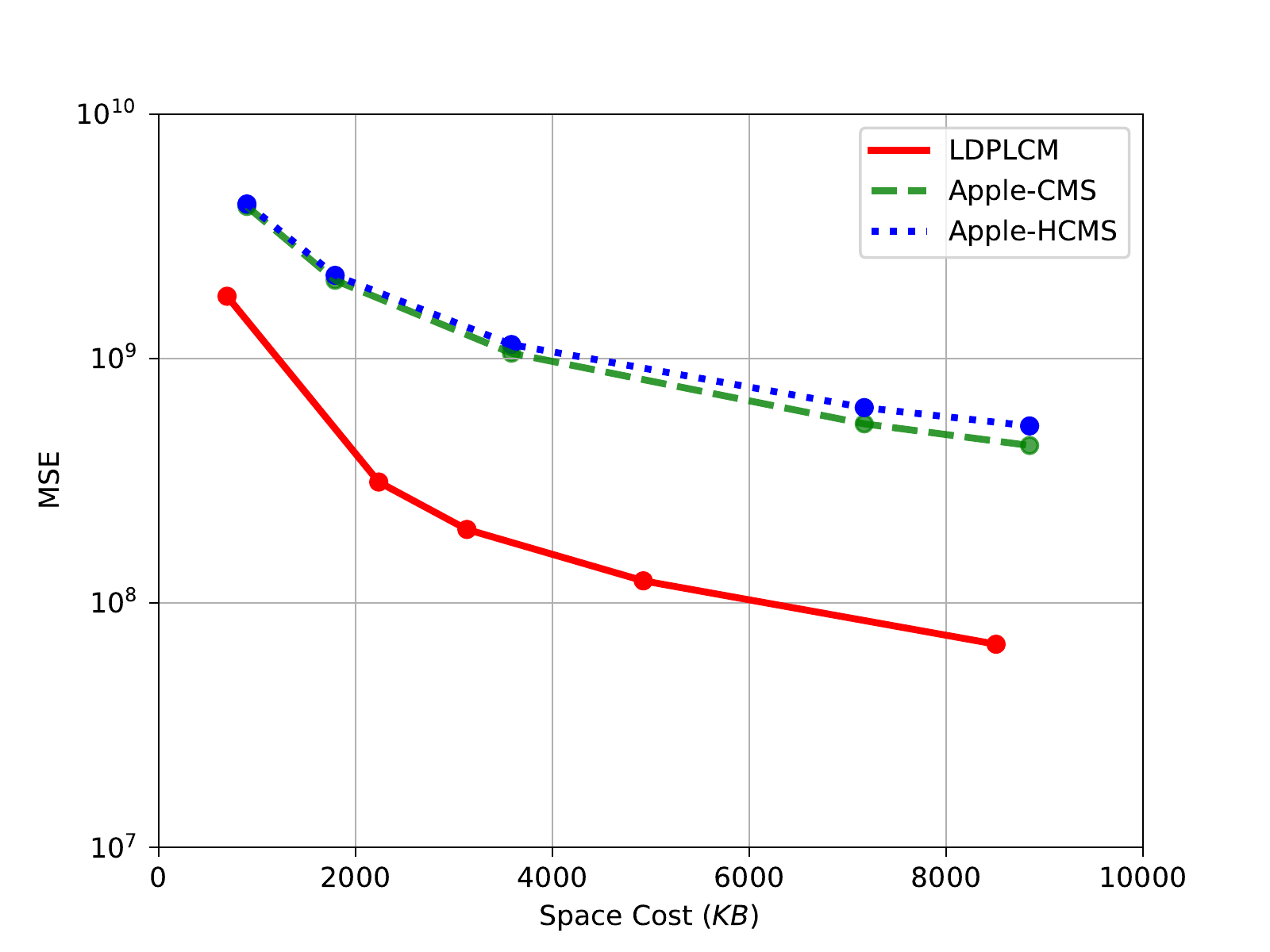}
\caption{The impact of space cost on the accuracy.}
\label{Fig:100zipf_spacecost}
\end{figure*}

\subsection{Efficiency}
We test the query efficiency of different methods on the Zipf(100M) dataset. We test 50,000 data records and average the query time as shown in Figure~\ref{Fig:100zipf_eff}. Our algorithm is slightly less efficient than the Apple-CMS, Apple-HCMS. The reason is that the estimation time of our method includes not only the time to estimate the frequency according to the sketch or matrix, but also the time to predict the frequency based on the model. FLH is the most efficient method, since it maps the data to a small domain $e^\epsilon+1$. But its accuracy is worse than the other methods as shown in Fig~\ref{Fig:accuracy}.%From the figure~\ref{Fig:100zipf_eff}, we can find that our method LDPLCM uses the frequency model to improve the accuracy without adding much extra query overhead.

\begin{figure*}
\centering
    \includegraphics[scale=0.4]{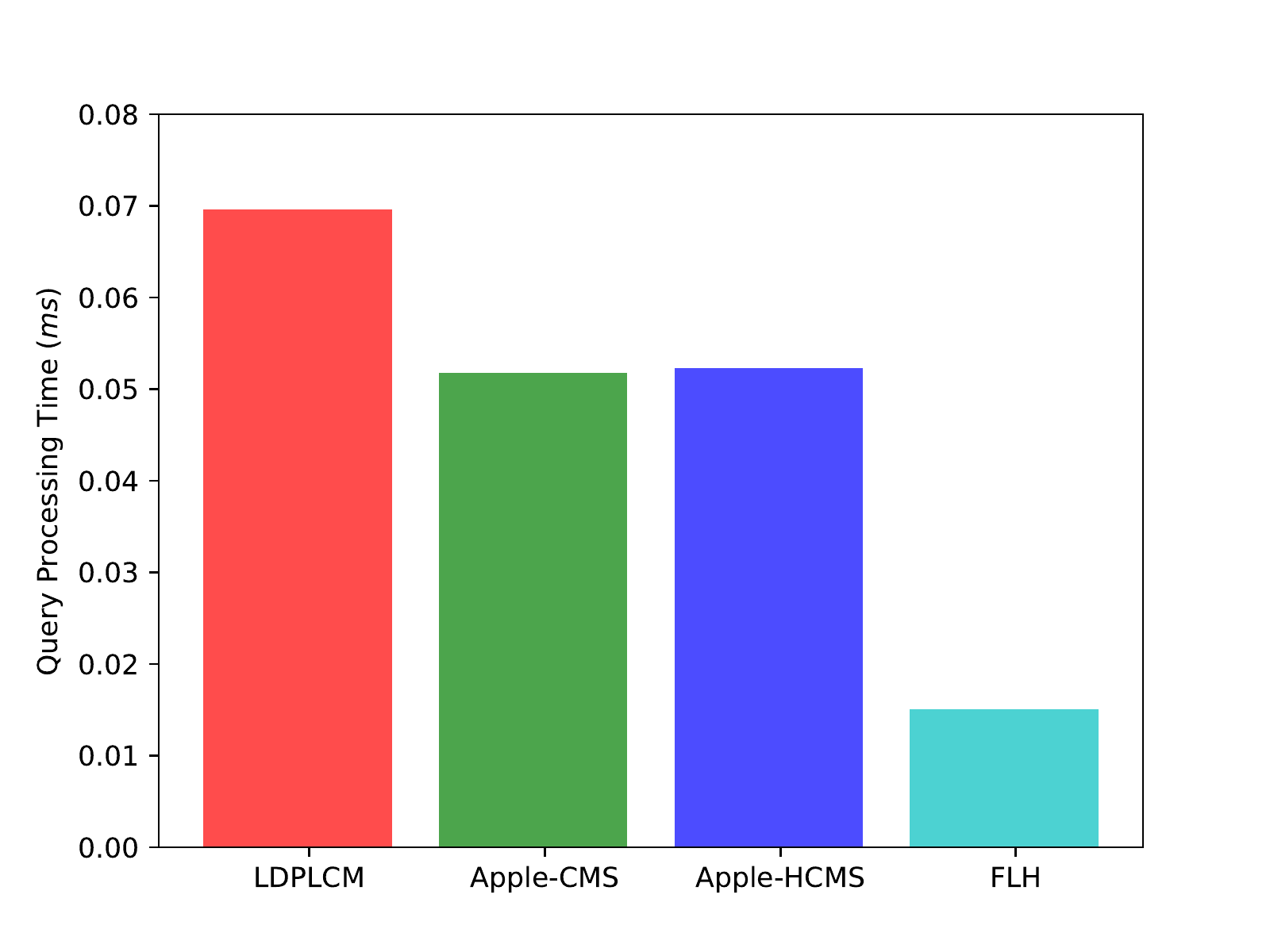}
\caption{Efficiency.}
\label{Fig:100zipf_eff}
\end{figure*}

\subsection{The impact of parameters}

We test the impact of different parameters including sketch parameters $(m,k)$, privacy budget $\epsilon$, sampling rate $r$, and the ratio of high-frequent items $\theta$ on the accuracy.%, including sketch parameters $(m,k)$, privacy budget $\epsilon$, sampling rate $r$ and the ratio of high-frequent items $\theta$ as shown in Figure \ref{Fig:10sketch parameters}-\ref{Fig:100sketch parameters}. The training model parameters are as same as in section 5.2.}

\subsubsection{The impact of sketch parameters ($m$, $k$) on the accuracy}

\begin{figure*}
\centering
  \subfigure[The impact of hash functions $k$.]{
    \label{Fig:10zipf_sketch(k)}
    \includegraphics[width=0.4\linewidth]{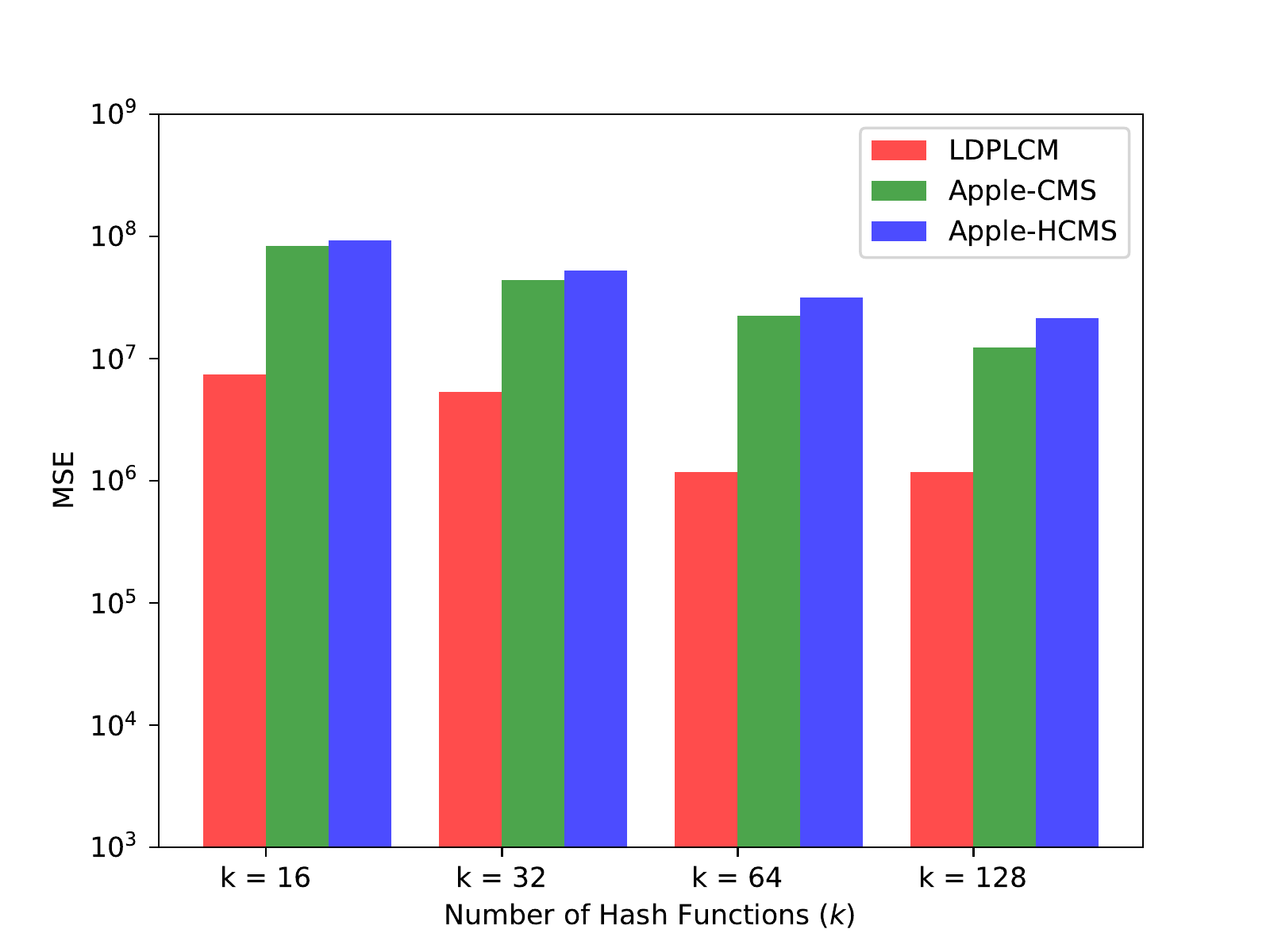}}
  \subfigure[The impact of sketch vectors $m$.]{
    \label{Fig:10zipf_sketch(m)}
    \includegraphics[width=0.4\linewidth]{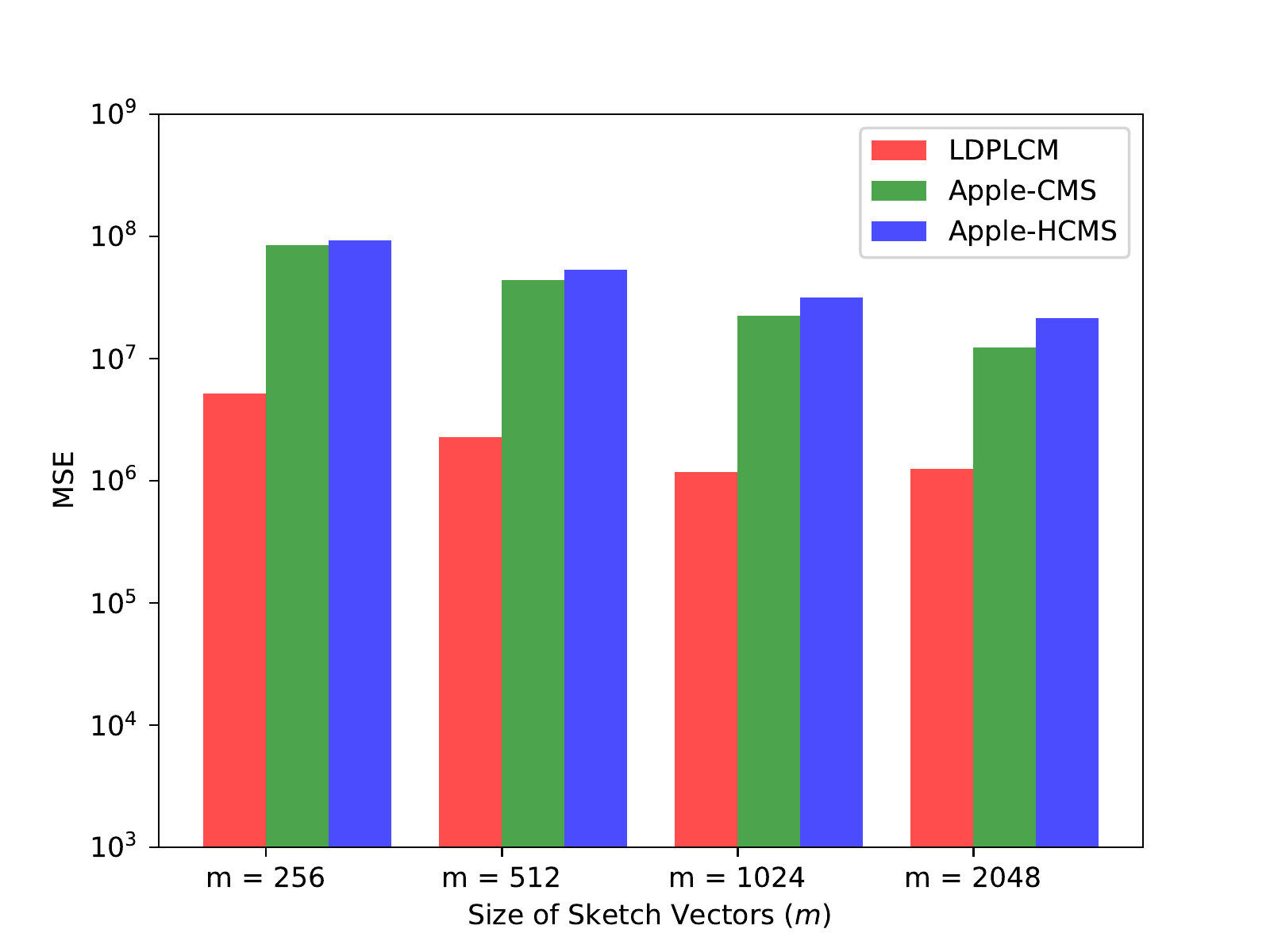}}
\caption{The impact of sketch parameters $(m,k)$ on accuracy (Zipf(10M) dataset).}
\label{Fig:10sketch parameters}
\end{figure*}

We test the impact of sketch parameters $(m,k)$ on the accuracy with Zipf(10M) and Zipf(100M) datasets. We fix $m=1024$ in Figure~\ref{Fig:10zipf_sketch(k)} and Figure~\ref{Fig:100zipf_sketch(k)} and $k=64$ in Figure~\ref{Fig:10zipf_sketch(m)} and Figure~\ref{Fig:100zipf_sketch(m)}. The other parameters are $\theta=0.5, r=0.1, \epsilon=4$. Also, we fix $m=512$ in Figure~\ref{Fig:wesad_sketch(k)} and $k=32$ in Figure~\ref{Fig:wesad_sketch(m)} while the other parameters are $\theta=0.4, r=0.1, \epsilon=4$. From the Figure~\ref{Fig:10sketch parameters}, Figure~\ref{Fig:100sketch parameters} and Figure~\ref{Fig:wesad parameters}, we can learn that LDPLCM outperforms Apple-CMS and Apple-HCMS as $(m,k)$ keeps increasing. The maximum gap is obtained when the $(m,k)$ takes the minimum value. This phenomenon is reasonable because the smaller sketches mean more hash collisions, and then the more significant the accuracy improvement of our method by avoiding the collisions. A widening gap between our method and Apple-CMS is more pronounced when the domain is larger since a larger domain brings more hash collisions when the sketch is fixed. LDPLCM uses the model to predict the estimation for high-frequent items and reduce the collisions so that it improves the accuracy for low-frequent items.

\subsubsection{The impact of privacy budget $\epsilon$ on the accuracy}
\begin{figure*}
\centering
  \subfigure[The impact of hash functions $k$.]{
    \label{Fig:100zipf_sketch(k)}
    \includegraphics[width=0.4\linewidth]{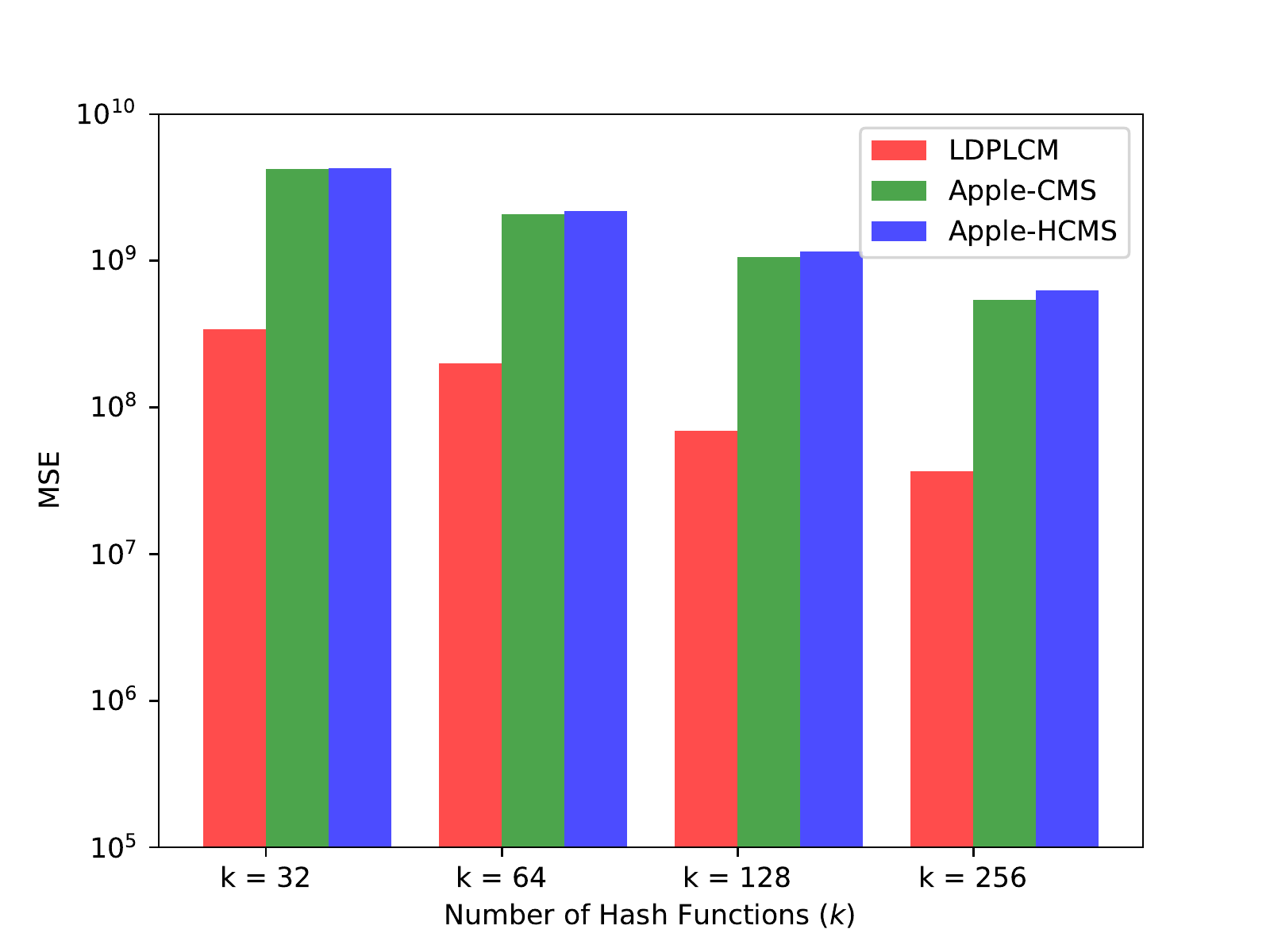}}
  \subfigure[The impact of sketch vectors $m$.]{
    \label{Fig:100zipf_sketch(m)}
    \includegraphics[width=0.4\linewidth]{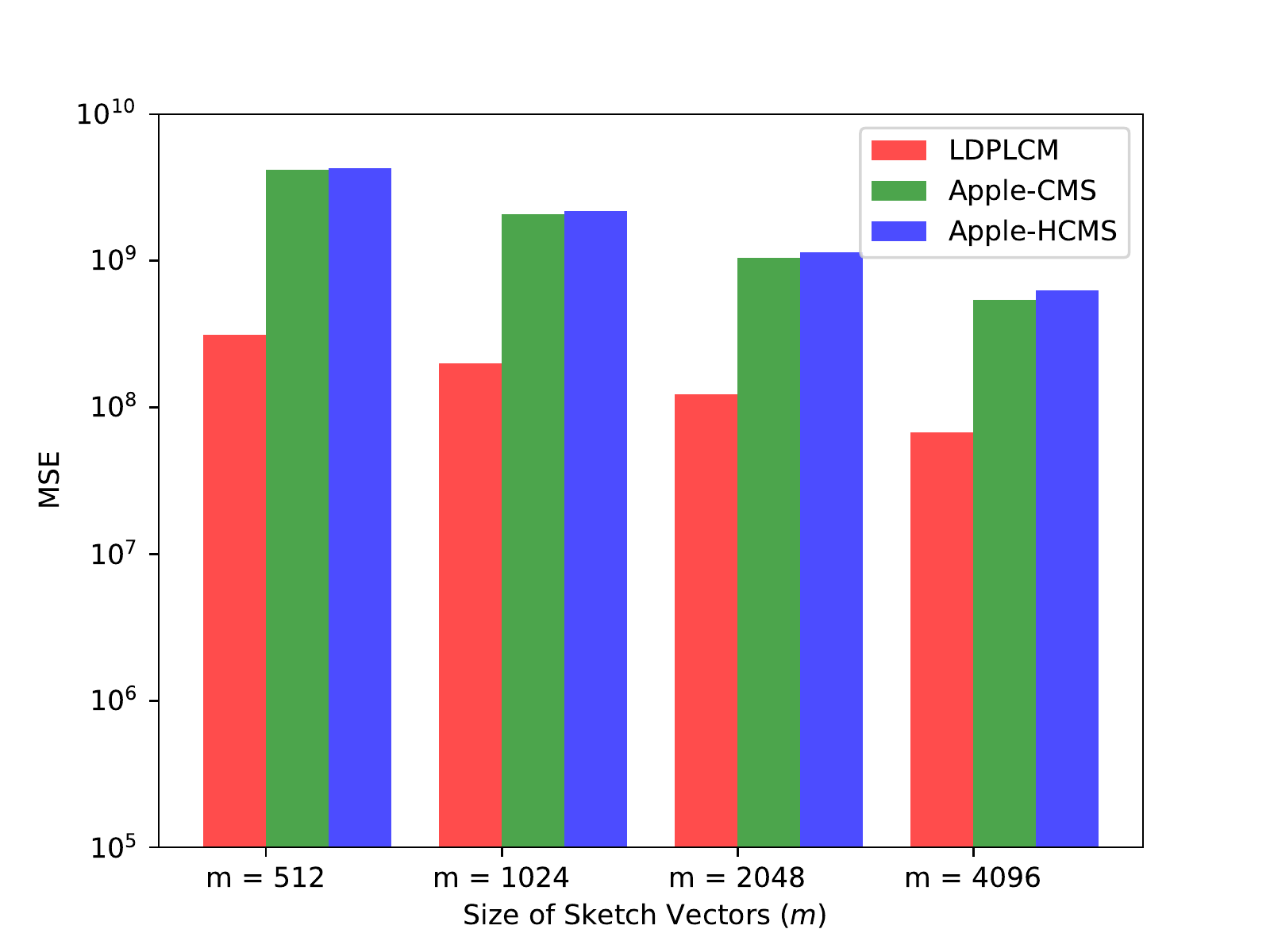}}
\caption{The impact of sketch parameters $(m,k)$ on the accuracy (Zipf(100M) dataset).}
\label{Fig:100sketch parameters}
\end{figure*}

\begin{figure*}
\centering
  \subfigure[The impact of hash functions $k$.]{
    \label{Fig:wesad_sketch(k)}
    \includegraphics[width=0.4\linewidth]{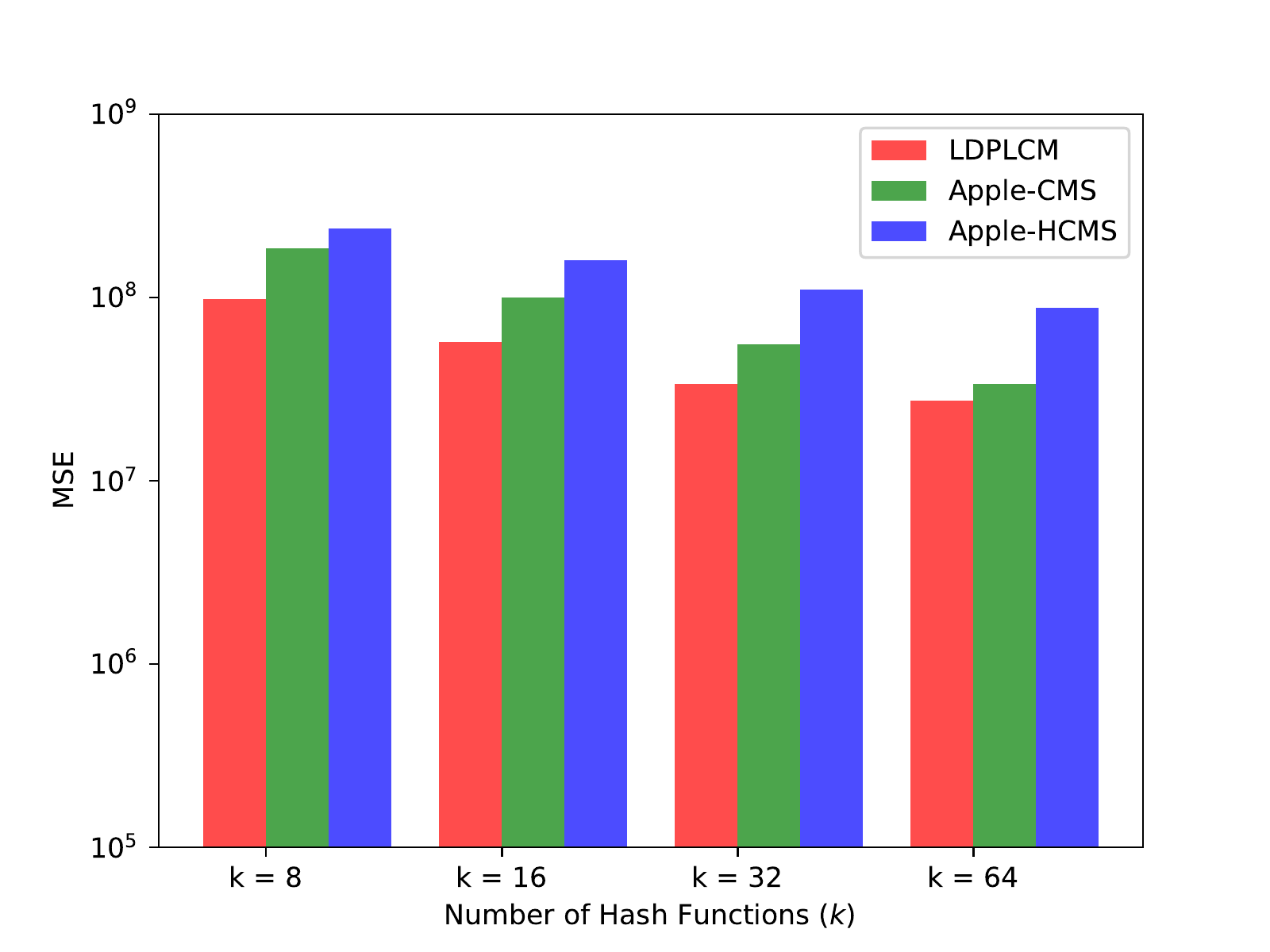}}
  \subfigure[The impact of sketch vectors $m$.]{
    \label{Fig:wesad_sketch(m)}
    \includegraphics[width=0.4\linewidth]{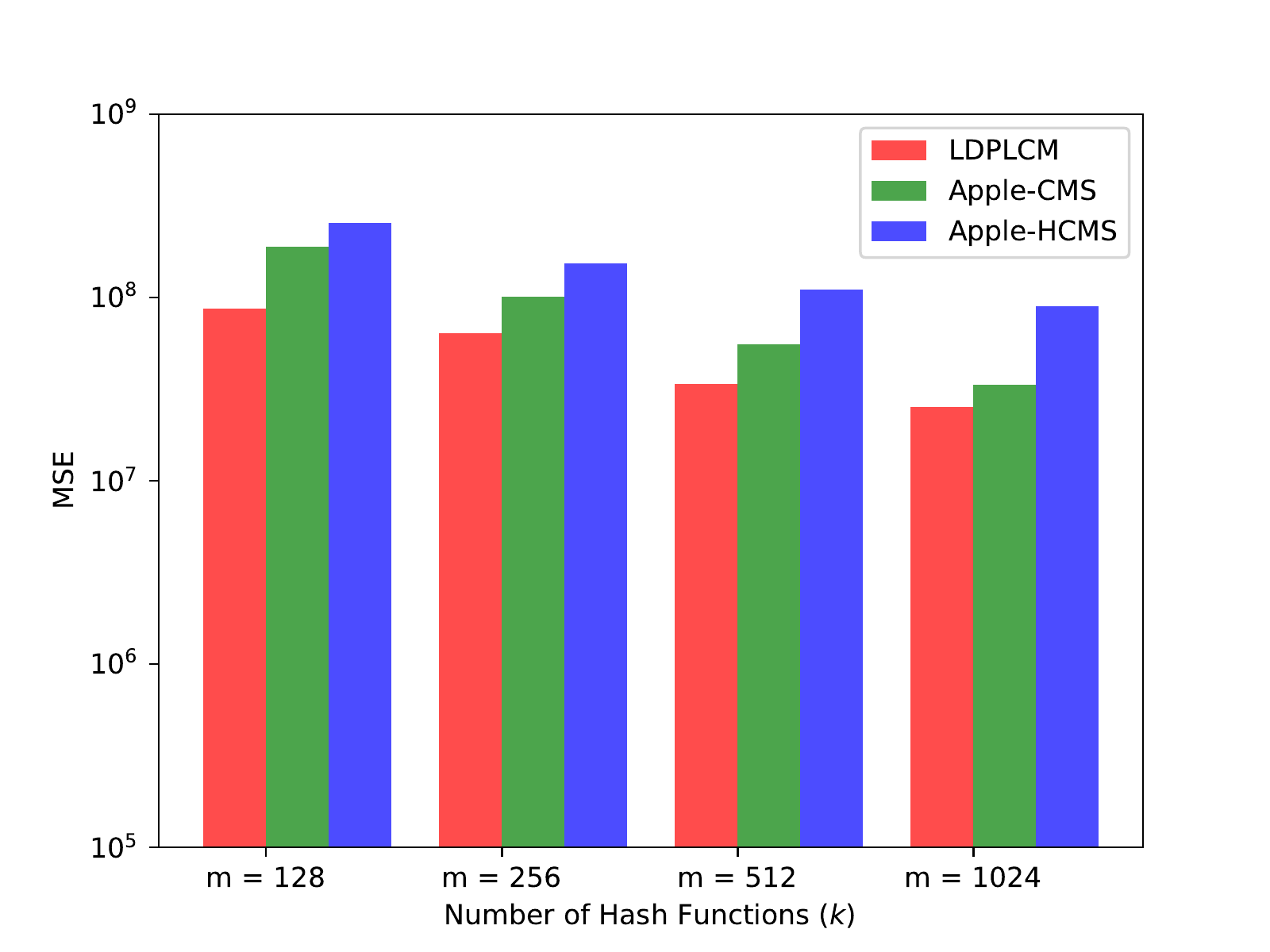}}
\caption{The impact of sketch parameters $(m,k)$ on the accuracy (Wesad dataset).}
\label{Fig:wesad parameters}
\end{figure*}

\begin{figure}
\centering
  \subfigure[Zipf(10M).]{
    \label{Fig:10zipf_eps}
    \includegraphics[width=0.3\linewidth]{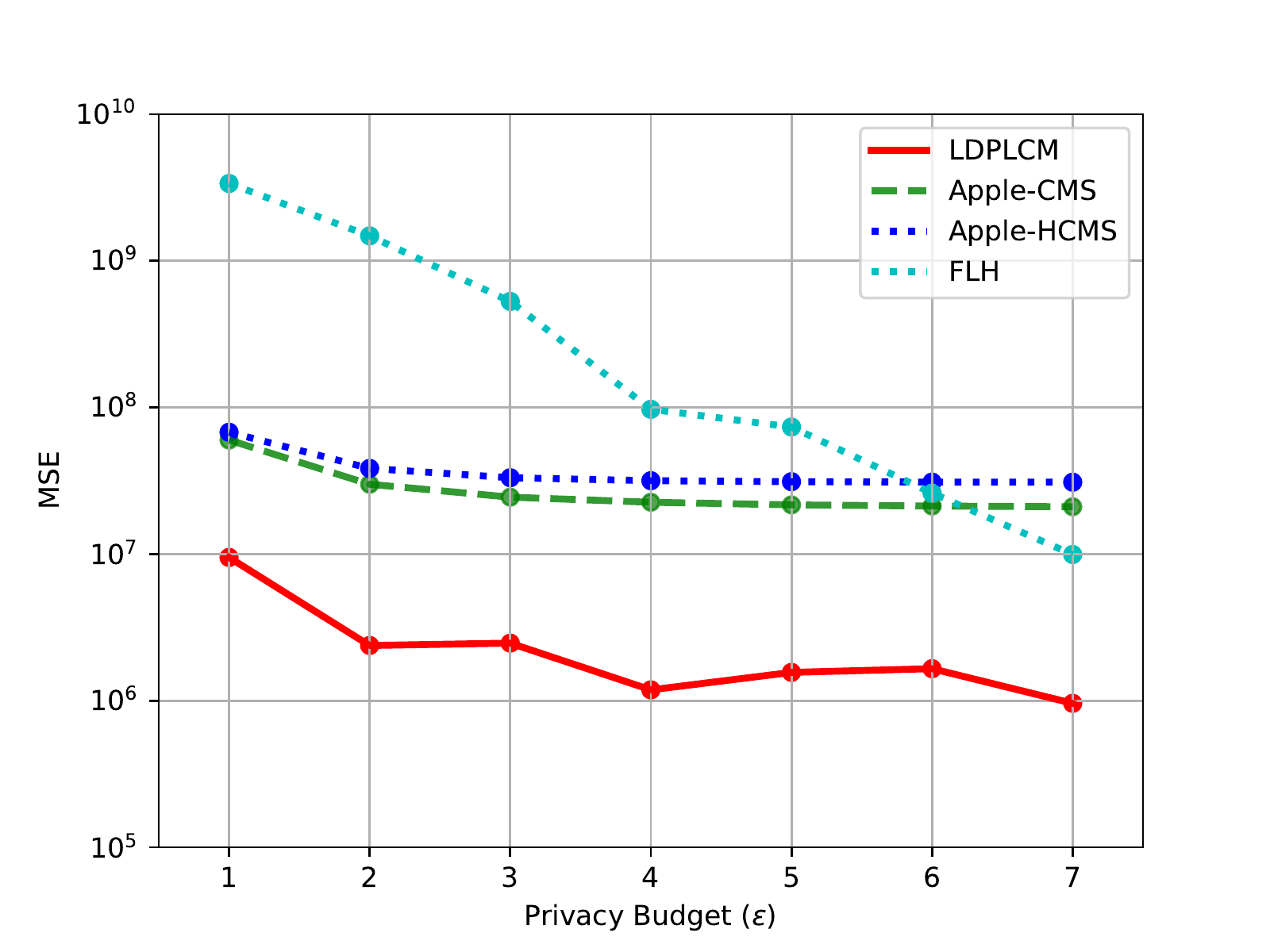}}
  \subfigure[Zipf(100M).]{
    \label{Fig:100zipf_eps}
    \includegraphics[width=0.3\linewidth]{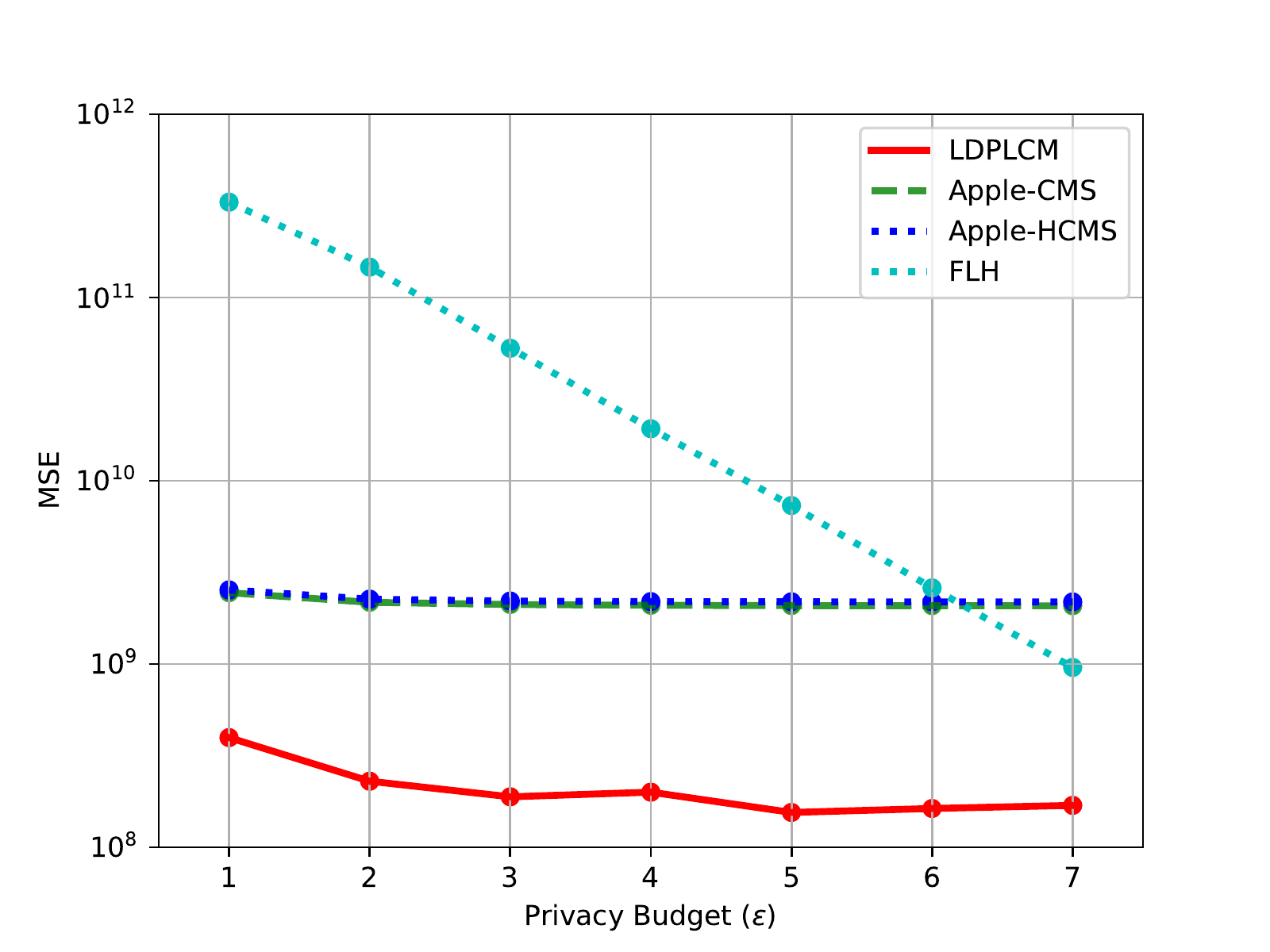}}
  \subfigure[Wesad.]{
    \label{Fig:wesad_eps}
    \includegraphics[width=0.3\linewidth]{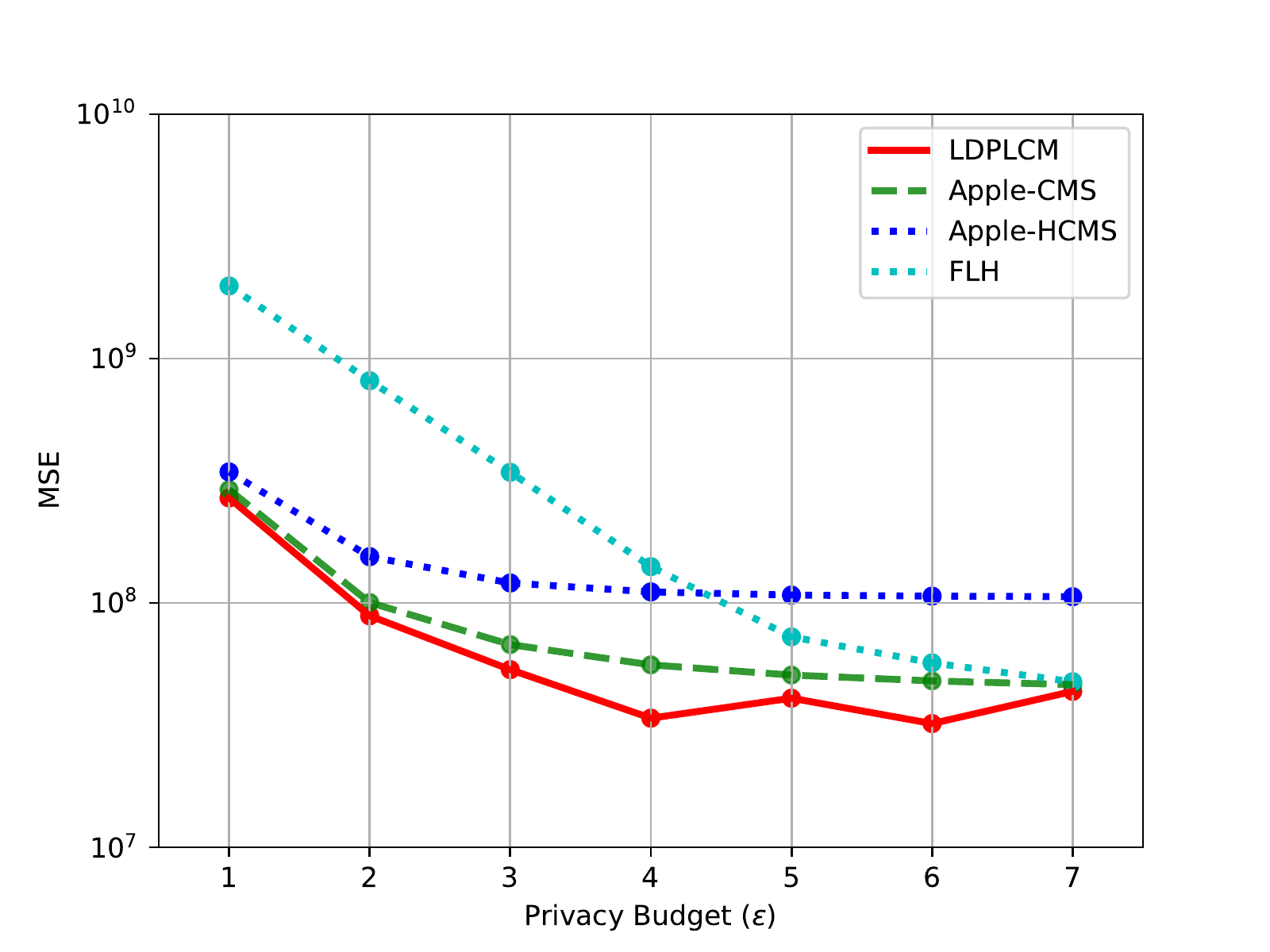}}
\caption{The impact of privacy budget $\epsilon$ on the accuracy.}
\label{Fig:dataset_eps}
\end{figure}

In Figure~\ref{Fig:dataset_eps}, we vary the privacy budget $\epsilon \in \left \{ 1,2,...,7 \right \} $ on the accuracy with the Zipf(10M), Zipf(100M) and WESAD datasets. For Zipf(10M) and Zipf(100M) datasets, we fix $k=64, m=1024, \theta=0.5, r=0.1$. For WESAD dataset, we fix $k=32, m=512, \theta=0.4, r=0.1$. We also fix $k'=128$ in FLH for all datasets. We can see that our algorithm is more accurate than Apple-CMS, Apple-HCMS and FLH with different privacy budgets on different datasets. The reason is that our method reduces the estimation errors of low-frequent items by avoiding hash collisions between high-frequent and low-frequent items, while the other methods face the errors caused by a large number of hash collisions in addition to the noises introduced by LDP. By comparing the Fig~\ref{Fig:10zipf_eps}, Fig~\ref{Fig:100zipf_eps} and Fig~\ref{Fig:wesad_eps}, we can learn that our method outperforms the other methods on datasets with a large domain and small privacy budget.

% The reason is that LDPLCM provides more accurate frequency estimation for the low-frequent items, which results in the smaller error among the data domain. Also, the trend of total MSE is decreased in Apple-CMS and not in LDPLCM. Since LDPLCM takes the frequency model to distinguish items in queries, MSE will fluctuate while predicting the high-frequent items so that the same situation occurs in total MSE. It means that fitting data close to the true frequency is necessary.

\subsubsection{The impact of sampling rate $r$ on the accuracy}

We set $k=64$, $m=1024$, $\theta=0.5$ and $\epsilon=4$ in Figure~\ref{Fig:100zipf_sample_rate}. We vary the sampling rate $r \in \left \{ 0.10,0.15,...,0.30 \right \}$ on the accuracy with the Zipf(100M) dataset. To ensure fairness, we make the precision of each frequency model similar for different sampling rates. We can find that the error decreases as the sampling rate increases. This phenomenon is consistent with the theoretical result, because the higher sampling rate reduces the error of sampling-based estimation for high-frequent data items. %The learned-model provides more accurate estimation and decreases the hash collisions in the estimating queries.}

%\begin{figure}
%\centering
%    \includegraphics[scale=0.4]{Figure_11}
%\caption{The impact of sampling rate $r$ on the accuracy of LDPLCM.}
%\label{Fig:100zipf_sample_rate}
%\end{figure}

\begin{figure}[htbp]
\centering
\begin{minipage}[t]{0.48\textwidth}
\centering
\includegraphics[scale=0.4]{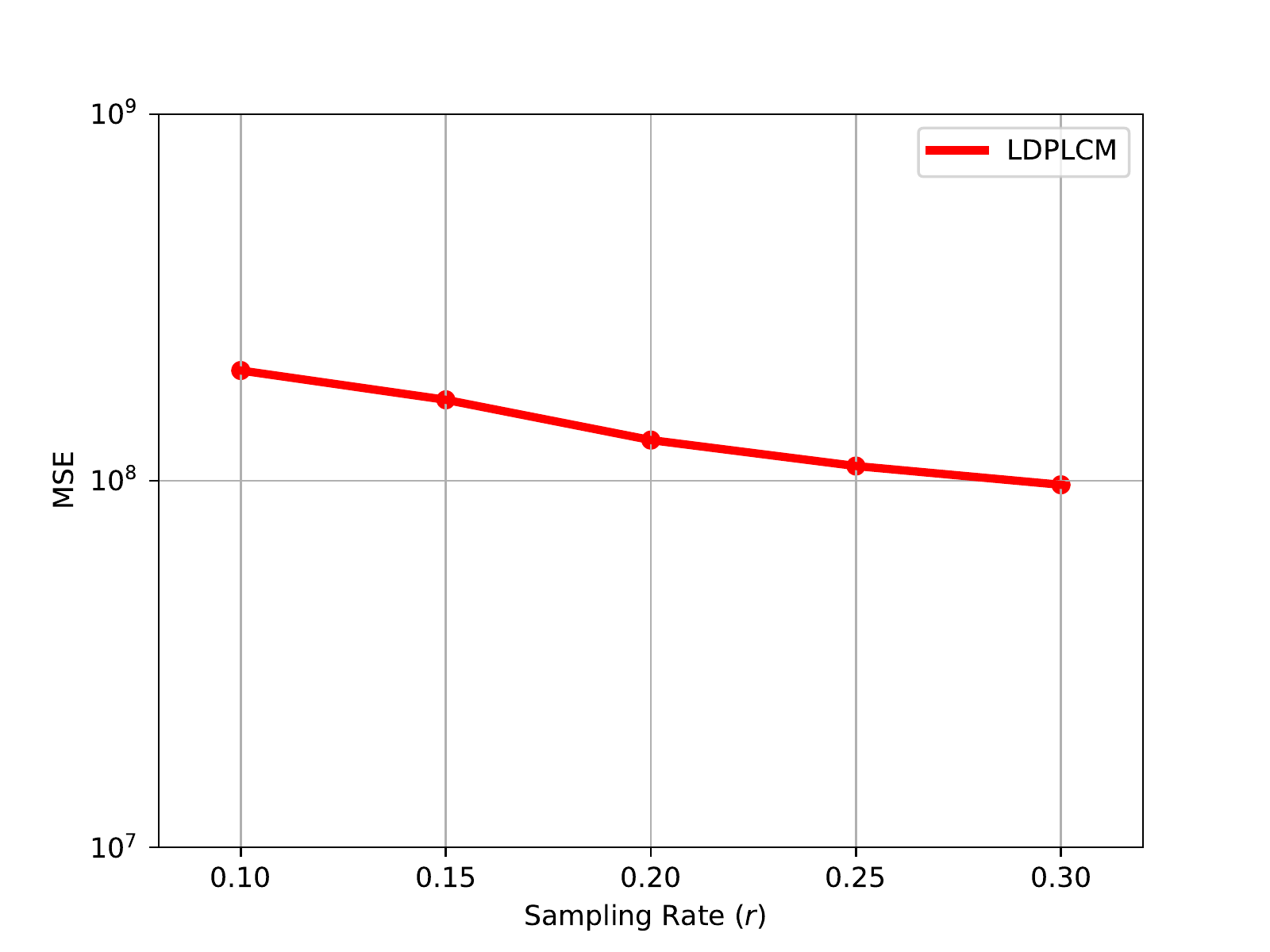}
\caption{The impact of sampling rate $r$ on the accuracy of LDPLCM.}
\label{Fig:100zipf_sample_rate}
\end{minipage}
\begin{minipage}[t]{0.48\textwidth}
\centering
\includegraphics[scale=0.4]{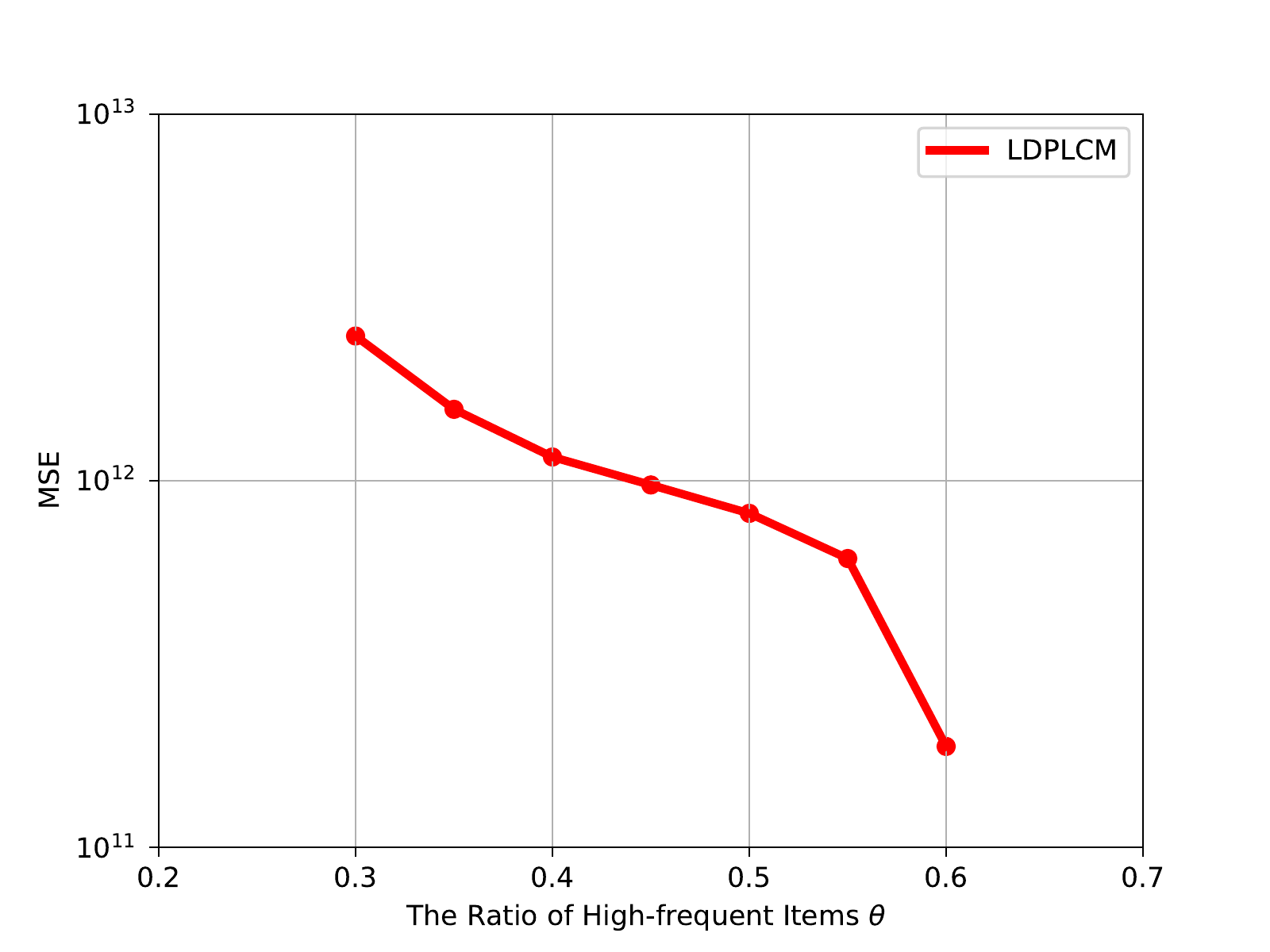}
\caption{The impact of $\theta$ on the accuracy of LDPLCM.}
\label{Fig:100zipf_theta}
\end{minipage}
\end{figure}

\subsubsection{The impact of $\theta$ on the accuracy}

%\begin{figure}
%\centering
%    \includegraphics[scale=0.4]{Figure_12}
%\caption{The impact of $\theta$ on the accuracy of LDPLCM.}
%\label{Fig:100zipf_theta}
%\end{figure}

The parameter $\theta$ means the ratio of the total frequencies of high-frequent items to the frequencies of all the items. We vary $\theta \in \{0.3, 0.35, ..., 0.6\}$ and test the impact of $\theta$ on the accuracy with the Zipf(100M) dataset. We set the parameters as $k=64$, $m=1024$, $r=0.1$, and $\epsilon=4$. We can learn from Figure~\ref{Fig:100zipf_theta} that a higher $\theta$ makes LDPLCM more accurate. The reason is that a higher $\theta$ means more high-frequent items are predicted by the model instead of the sketch. As a result, the hash collisions in the sketch are reduced, which leads to a more accurate estimation.
%Due to the higher theta, the frequency boundary is larger. More items are marked as high-frequent so that fewer data items stored in sketch. The reduction of total MSE is taken by improving the estimation accuracy of the low-frequent items.

\subsubsection{The impact of $s$ on the accuracy}

\begin{figure}
\centering
    \includegraphics[scale=0.4]{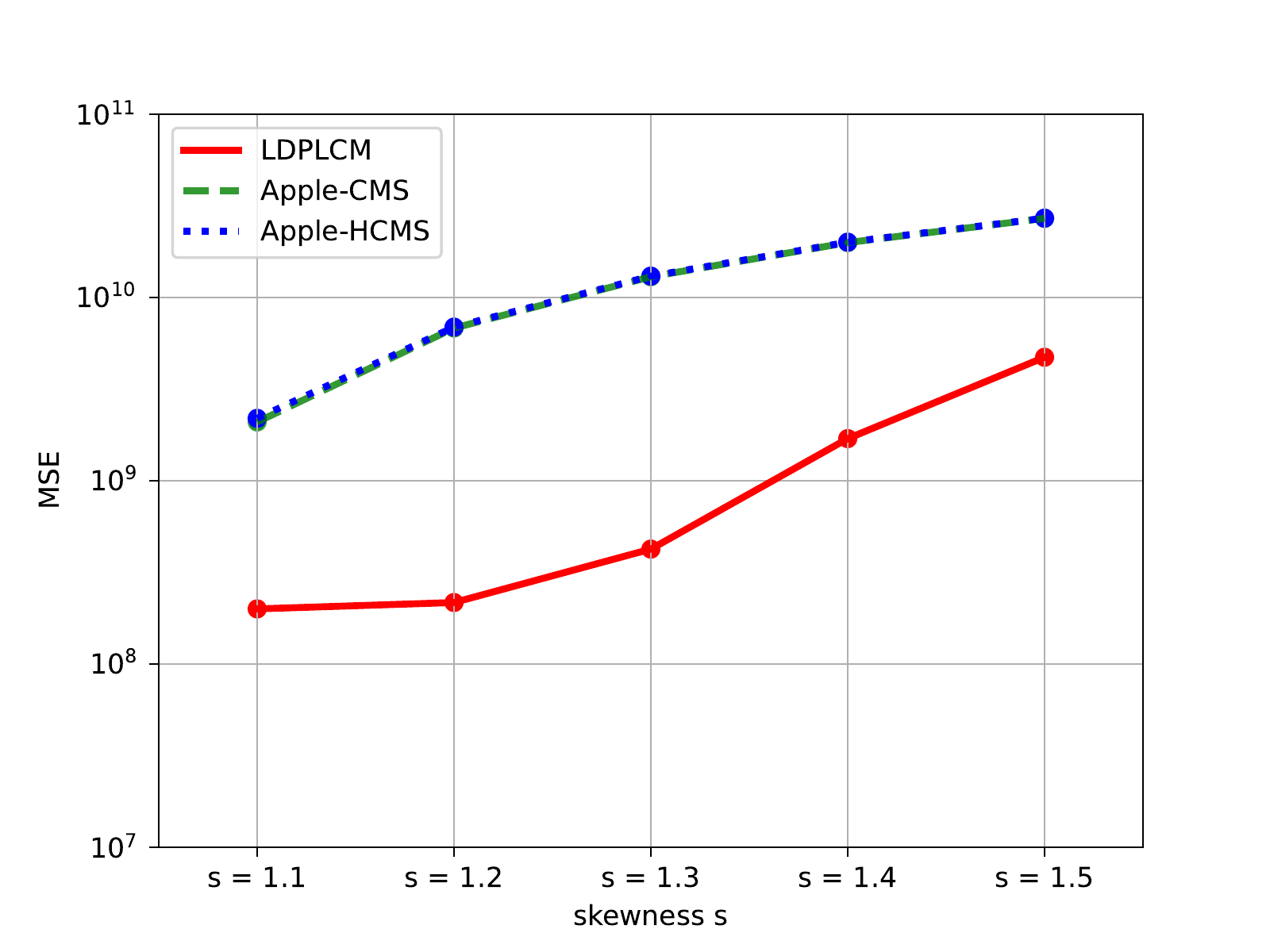}
\caption{The impact of $s$ on the accuracy of LDPLCM.}
\label{Fig:100zipf_s}
\end{figure}

The parameter $s$ means the skewness of the Zipf distribution dataset. We vary $s \in \{1.1, 1.2, ..., 1.5\}$ and test the impact of $s$ on the accuracy with the Zipf(100M) dataset. We set the parameters as $k=64$, $m=1024$, $r=0.1$, and $\epsilon=4$. As shown in Figure~\ref{Fig:100zipf_s}, we can find that a smaller $s$ makes LDPLCM more accurate. As the sum of frequency of high-frequent items in the experiments with different skewness ``$s$'' are the same, the lower skewness means more items will be regarded as high-frequent items. Thus, more items are predicted by the frequency model. As a result, it reduces more hash collisions and leads to a more accurate estimation.

\subsection{Summary of experimental results}

The experimental results are summarized as follows:

\begin{itemize}
\item[$\bullet$] LDPLCM is more accurate than Apple-CMS, Apple-HCMS and FLH for frequency estimation on datasets with large domains.

\item[$\bullet$] The higher the frequency of high-frequent items removed from the sketch, the more accurate the low-frequent items estimations become.

%\item[$\bullet$] The higher sampling rate $r$ improves the prediction ability of model and increases the accuracy for high-frequent data items.

\item[$\bullet$] The superiority of LDPLCM is more obvious when the sketch is limited to a small size.

\end{itemize}

\section{Conclusion}
In this paper, we propose the LDPLCM algorithm to reduce hash collisions and provide more accurate frequency estimations under LDP, especially for data with a larger domain. We train the frequency model to distinguish the high-frequent items and the low-frequent items, and we separate the storage of the items with different frequency properties without leaking privacy. Our method focuses on the frequency estimation for one-dimensional data under LDP. Extending the method for multidimensional data is still challenging. Some frontier works adopt the multidimensional histograms or grids for multidimensional frequency estimation under LDP, however, they still suffer from the curse of dimension and the trade-off between the accuracy and the utility. We will try to extend our work for multidimensional data in the future.

\section*{Acknowledgements}
This work was supported by NSFC grant 62202113, the Major Key Project of PCL (Grant No.PCL2021A09, PCL2021A02, PCL2022A03).



\bibliographystyle{elsarticle-num}
\bibliography{mybibfile}

\begin{thebibliography}{10}
\expandafter\ifx\csname url\endcsname\relax
  \def\url#1{\texttt{#1}}\fi
\expandafter\ifx\csname urlprefix\endcsname\relax\def\urlprefix{URL }\fi
\expandafter\ifx\csname href\endcsname\relax
  \def\href#1#2{#2} \def\path#1{#1}\fi

\bibitem{sweeney2002k}
L.~Sweeney, k-anonymity: A model for protecting privacy, International Journal
  of Uncertainty, Fuzziness and Knowledge-Based Systems 10~(05) (2002)
  557--570.

\bibitem{dwork2008differential}
C.~Dwork, Differential privacy: A survey of results, in: International
  conference on theory and applications of models of computation, Springer,
  2008, pp. 1--19.

\bibitem{kasiviswanathan2011can}
S.~P. Kasiviswanathan, H.~K. Lee, K.~Nissim, S.~Raskhodnikova, A.~Smith, What
  can we learn privately?, SIAM Journal on Computing 40~(3) (2011) 793--826.

\bibitem{erlingsson2014rappor}
{\'U}.~Erlingsson, V.~Pihur, A.~Korolova, Rappor: Randomized aggregatable
  privacy-preserving ordinal response, in: Proceedings of the 2014 ACM SIGSAC
  conference on computer and communications security, 2014, pp. 1054--1067.

\bibitem{fanti2015building}
G.~Fanti, V.~Pihur, {\'U}.~Erlingsson, Building a rappor with the unknown:
  Privacy-preserving learning of associations and data dictionaries, arXiv
  preprint arXiv:1503.01214.

\bibitem{APPLE}
A.~Differential Privacy~Team, Learning with privacy at scale.

\bibitem{ding2017collecting}
B.~Ding, J.~Kulkarni, S.~Yekhanin, Collecting telemetry data privately,
  Advances in Neural Information Processing Systems 30.

\bibitem{wang2017locally}
T.~Wang, J.~Blocki, N.~Li, S.~Jha, Locally differentially private protocols for
  frequency estimation, in: 26th USENIX Security Symposium (USENIX Security
  17), 2017, pp. 729--745.

\bibitem{yildirim2020differentially}
S.~Y{\i}ld{\i}r{\i}m, K.~Kaya, S.~Ayd{\i}n, H.~B. Erentu{\u{g}}, Differentially
  private frequency sketches for intermittent queries on large data streams,
  in: 2020 IEEE International Conference on Big Data (Big Data), IEEE, 2020,
  pp. 4083--4092.

\bibitem{li2020wavingsketch}
J.~Li, Z.~Li, Y.~Xu, S.~Jiang, T.~Yang, B.~Cui, Y.~Dai, G.~Zhang, Wavingsketch:
  An unbiased and generic sketch for finding top-k items in data streams, in:
  Proceedings of the 26th ACM SIGKDD International Conference on Knowledge
  Discovery \& Data Mining, 2020, pp. 1574--1584.

\bibitem{karp2003simple}
R.~M. Karp, S.~Shenker, C.~H. Papadimitriou, A simple algorithm for finding
  frequent elements in streams and bags, ACM Transactions on Database Systems
  (TODS) 28~(1) (2003) 51--55.

\bibitem{tai2018sketching}
K.~S. Tai, V.~Sharan, P.~Bailis, G.~Valiant, Sketching linear classifiers over
  data streams, in: Proceedings of the 2018 International Conference on
  Management of Data, 2018, pp. 757--772.

\bibitem{cormode2003finding}
G.~Cormode, F.~Korn, S.~Muthukrishnan, D.~Srivastava, Finding hierarchical
  heavy hitters in data streams, in: Proceedings 2003 VLDB Conference,
  Elsevier, 2003, pp. 464--475.

\bibitem{tang2016graph}
N.~Tang, Q.~Chen, P.~Mitra, Graph stream summarization: From big bang to big
  crunch, in: Proceedings of the 2016 International Conference on Management of
  Data, 2016, pp. 1481--1496.

\bibitem{yang2019adaptive}
T.~Yang, J.~Jiang, P.~Liu, Q.~Huang, J.~Gong, Y.~Zhou, R.~Miao, X.~Li,
  S.~Uhlig, Adaptive measurements using one elastic sketch, IEEE/ACM
  Transactions on Networking 27~(6) (2019) 2236--2251.

\bibitem{liu2016one}
Z.~Liu, A.~Manousis, G.~Vorsanger, V.~Sekar, V.~Braverman, One sketch to rule
  them all: Rethinking network flow monitoring with univmon, in: Proceedings of
  the 2016 ACM SIGCOMM Conference, 2016, pp. 101--114.

\bibitem{basat2019randomized}
R.~B. Basat, X.~Chen, G.~Einziger, R.~Friedman, Y.~Kassner, Randomized
  admission policy for efficient top-k, frequency, and volume estimation,
  IEEE/ACM Transactions on Networking 27~(4) (2019) 1432--1445.

\bibitem{kairouz2016discrete}
P.~Kairouz, K.~Bonawitz, D.~Ramage, Discrete distribution estimation under
  local privacy, in: International Conference on Machine Learning, PMLR, 2016,
  pp. 2436--2444.

\bibitem{bloom1970space}
B.~H. Bloom, Space/time trade-offs in hash coding with allowable errors,
  Communications of the ACM 13~(7) (1970) 422--426.

\bibitem{duchi2013local}
J.~C. Duchi, M.~I. Jordan, M.~J. Wainwright, Local privacy and statistical
  minimax rates, in: 2013 IEEE 54th Annual Symposium on Foundations of Computer
  Science, IEEE, 2013, pp. 429--438.

\bibitem{wang2019locally}
T.~Wang, M.~Lopuha{\"a}-Zwakenberg, Z.~Li, B.~Skoric, N.~Li, Locally
  differentially private frequency estimation with consistency, arXiv preprint
  arXiv:1905.08320.

\bibitem{wang2019answering}
T.~Wang, B.~Ding, J.~Zhou, C.~Hong, Z.~Huang, N.~Li, S.~Jha, Answering
  multi-dimensional analytical queries under local differential privacy, in:
  Proceedings of the 2019 International Conference on Management of Data, 2019,
  pp. 159--176.

\bibitem{murakami2019utility}
T.~Murakami, Y.~Kawamoto, $\{$Utility-Optimized$\}$ local differential privacy
  mechanisms for distribution estimation, in: 28th USENIX Security Symposium
  (USENIX Security 19), 2019, pp. 1877--1894.

\bibitem{wei2020asgldp}
C.~Wei, S.~Ji, C.~Liu, W.~Chen, T.~Wang, Asgldp: collecting and generating
  decentralized attributed graphs with local differential privacy, IEEE
  Transactions on Information Forensics and Security 15 (2020) 3239--3254.

\bibitem{xu2020collecting}
M.~Xu, B.~Ding, T.~Wang, J.~Zhou, Collecting and analyzing data jointly from
  multiple services under local differential privacy, Proceedings of the VLDB
  Endowment 13~(12) (2020) 2760--2772.

\bibitem{jia2019calibrate}
J.~Jia, N.~Z. Gong, Calibrate: Frequency estimation and heavy hitter
  identification with local differential privacy via incorporating prior
  knowledge, in: IEEE INFOCOM 2019-IEEE Conference on Computer Communications,
  IEEE, 2019, pp. 2008--2016.

\bibitem{cormode2021frequency}
G.~Cormode, S.~Maddock, C.~Maple, Frequency estimation under local differential
  privacy, Proceedings of the VLDB Endowment 14~(11) (2021) 2046--2058.

\bibitem{charikar2002finding}
M.~Charikar, K.~Chen, M.~Farach-Colton, Finding frequent items in data streams,
  in: International Colloquium on Automata, Languages, and Programming,
  Springer, 2002, pp. 693--703.

\bibitem{cormode2005improved}
G.~Cormode, S.~Muthukrishnan, An improved data stream summary: the count-min
  sketch and its applications, Journal of Algorithms 55~(1) (2005) 58--75.

\bibitem{roy2016augmented}
P.~Roy, A.~Khan, G.~Alonso, Augmented sketch: Faster and more accurate stream
  processing, in: Proceedings of the 2016 International Conference on
  Management of Data, 2016, pp. 1449--1463.

\bibitem{estan2002new}
C.~Estan, G.~Varghese, New directions in traffic measurement and accounting,
  in: Proceedings of the 2002 conference on Applications, technologies,
  architectures, and protocols for computer communications, 2002, pp. 323--336.

\bibitem{zhou2018cold}
Y.~Zhou, T.~Yang, J.~Jiang, B.~Cui, M.~Yu, X.~Li, S.~Uhlig, Cold filter: A
  meta-framework for faster and more accurate stream processing, in:
  Proceedings of the 2018 International Conference on Management of Data, 2018,
  pp. 741--756.

\bibitem{yang2018heavyguardian}
T.~Yang, J.~Gong, H.~Zhang, L.~Zou, L.~Shi, X.~Li, Heavyguardian: Separate and
  guard hot items in data streams, in: Proceedings of the 24th ACM SIGKDD
  International Conference on Knowledge Discovery \& Data Mining, 2018, pp.
  2584--2593.

\bibitem{qin2016heavy}
Z.~Qin, Y.~Yang, T.~Yu, I.~Khalil, X.~Xiao, K.~Ren, Heavy hitter estimation
  over set-valued data with local differential privacy, in: Proceedings of the
  2016 ACM SIGSAC Conference on Computer and Communications Security, 2016, pp.
  192--203.

\bibitem{kraska2018case}
T.~Kraska, A.~Beutel, E.~H. Chi, J.~Dean, N.~Polyzotis, The case for learned
  index structures, in: Proceedings of the 2018 international conference on
  management of data, 2018, pp. 489--504.

\bibitem{hsu2019learning}
C.-Y. Hsu, P.~Indyk, D.~Katabi, A.~Vakilian, Learning-based frequency
  estimation algorithms., in: International Conference on Learning
  Representations, 2019.

\bibitem{zhang2020learned}
M.~Zhang, H.~Wang, J.~Li, H.~Gao, Learned sketches for frequency estimation,
  Information Sciences 507 (2020) 365--385.

\bibitem{zhou2019rl}
Z.~Zhou, D.~Zhang, X.~Hong, Rl-sketch: Scaling reinforcement learning for
  adaptive and automate anomaly detection in network data streams, in: 2019
  IEEE 44th Conference on Local Computer Networks (LCN), IEEE, 2019, pp.
  340--347.

\bibitem{schmidt2018introducing}
P.~Schmidt, A.~Reiss, R.~Duerichen, C.~Marberger, K.~Van~Laerhoven, Introducing
  wesad, a multimodal dataset for wearable stress and affect detection, in:
  Proceedings of the 20th ACM international conference on multimodal
  interaction, 2018, pp. 400--408.

\end{thebibliography}
\end{document}